\documentclass[%
reprint,
 citeautoscript,
twocolumn,
 amsmath,
 amssymb,
 longbibliography,
pre,
floatfix,
]{revtex4-2}
\usepackage[utf8]{inputenc}
\usepackage{amsmath}
\usepackage{amsthm}
\usepackage{amsfonts}
\usepackage{amssymb}
\usepackage{graphicx}
\usepackage{numprint}
\usepackage{mathtools}
\usepackage[colorlinks=true,linkcolor=blue,citecolor=blue,urlcolor=cyan,bookmarks,breaklinks=true]{hyperref}
\usepackage{enumitem}
\usepackage{graphicx}% Include figure files
\usepackage{dcolumn}% Align table columns on decimal point
\usepackage{bm}% bold math
\usepackage{xcolor}
\usepackage[caption=false]{subfig}
\usepackage{comment}
\newtheorem{thm}{Theorem}

\begin{document}

\title{Can One Hear the Shape of a Crystal?}
\author{Haina Wang}
\affiliation{Department of Chemistry, Princeton University, Princeton,
NJ 08544}

\author{Salvatore Torquato}
\affiliation{Department of Chemistry, Princeton University, Princeton,
NJ 08544}
\affiliation{Department of Physics, Princeton University, Princeton, NJ 08544}
\affiliation{Princeton Materials Institute, Princeton University, Princeton, NJ 08540}
\affiliation{Program in Applied and Computational Mathematics, Princeton, NJ 08544}
\thanks{
Email:
\href{mailto:torquato@electron.princeton.edu}{torquato@electron.princeton.edu}}

\begin{abstract}
Isospectrality is a general fundamental concept often involving whether various operators can have identical spectra, i.e., the same set of eigenvalues. 
In the context of the Laplacian operator, the famous question ``Can one hear the shape of a drum?'' concerns whether different shaped drums can have the same vibrational modes.
The isospectrality of a lattice in $d$-dimensional Euclidean space $\mathbb{R}^d$ is a tantamount to whether it is uniquely determined by its theta series, i.e., the radial distribution function $g_2(r)$.
While much is known about the isospectrality of Bravais lattices across dimensions, little is known about this question of more general crystal (periodic) structures with an $n$-particle basis ($n \ge 2$). 
Here, we ask, What is $n_{\text{min}}(d)$, the minimum value of $n$ for inequivalent (i.e., unrelated by isometric symmetries) crystals with the same theta function in space dimension $d$? 
To answer these questions, we use rigorous methods as well as a precise numerical algorithm that enables us to determine the minimum multi-particle basis of inequivalent isospectral crystals.
Our algorithm identifies isospectral 4-, 3- and 2-particle bases in one, two and three spatial dimensions, respectively. 
For many of these isospectral crystals, we rigorously show that they indeed possess identical $g_2(r)$ up to infinite $r$.
Based on our analyses, we conjecture that $n_{\text{min}}(d) = 4, 3, 2$ for $d = 1, 2, 3$, respectively.
The identification of isospectral crystals enables one to study the degeneracy of the ground-state under the action of isotropic pair potentials.
Indeed, using inverse statistical-mechanical techniques, we find an isotropic pair potential whose low-temperature configurations in two dimensions obtained via simulated annealing can lead to both of two isospectral crystal structures with $n = 3$, the proportion of which can be controlled by the cooling rate.
Our findings provide general insights into the structural and ground-state degeneracies of crystal structures  as determined
by radial pair information.

\end{abstract}
\maketitle

\newpage
\section{Introduction}
\label{sec:intro}
Isospectrality is a general fundamental concept often involving whether various operators can have identical spectra, (i.e., the same set of eigenvalues counting multiplicities \cite{Da95}), and arises in a broad range of contexts, including classical mechanics \cite{Gi10, Ka66b}, photonics \cite{Yu16, Yu21}, quantum chaos \cite{Gi10, Ri81}, graph theory \cite{Bi93}, dynamical systems \cite{Ga04c} and computational chemistry \cite{He75}. 
Given a linear operator and system-specific boundary conditions, two Riemannian manifolds are said to be isospectral if the multiset (i.e., set with multiplicities counted) of their eigenvalues of the operator coincide.
In the context of the Laplacian operator $\Delta$, Kac \cite{Ka66b} posed this problem as the famous question ``Can one hear the shape of a drum?'', i.e., whether different shaped drums can have the same vibrational modes.
It has been established that this is not the case, as examples of isospectral drums have been identified across dimensions; see Ref. \cite{Mi64, Go92, Bu94}.

The spectrum of the Laplacian operator of the flat torus (which is a quotient of Euclidean space $\mathbb{R}^d$  by a Bravais lattice $\Lambda$) is directly related to the \textit{theta series} of the lattice, which encodes radial pair distance coordination information that exactly determines the radial distribution function $g_2(r)$ \cite{Mi64}; see Sec. \ref{sec:lattice} for precise definitions.
A Bravais lattice in $\mathbb{R}^d$ is a subgroup consisting of the integer linear combinations of vectors that constitute a basis for $\mathbb{R}^d$ \cite{To10d}.
In a Bravais lattice $\Lambda$, the space $\mathbb{R}^d$ can be geometrically divided into identical regions $F$ called fundamental cells, each of which contains one point.
Henceforth, we use the term lattice to refer to a Bravais lattice.
The isospectrality of a lattice in $\mathbb{R}^d$ is tantamount to whether it is uniquely determined by its theta series. 
In 1964, Milnor \cite{Mi64} first identified examples of two 16-dimensional (16D) flat tori that possess identical spectra but correspond to different lattices. 
Since then, isospectral lattices have been identified in as low as four dimensions \cite{Sc90}.
Furthermore, it has been established by Schiemann \cite{Sc90b} that isospectral lattices cannot exist in $d \leq 3$ dimensions.
(For a simple proof in the 2D case, see Ref. \cite{Co97}.)
However, while much is known about isospectrality of lattices, little is known about the ``isospectrality'' of more general crystal (periodic) structures, i.e., periodic point patterns with an $n$-particle basis with $n\geq 2$.

It is noteworthy that the isospectral problem for crystals is closely related to the structural degeneracy problem of pair statistics, which frequently emerge in chemical physics and statistical mechanics \cite{Le75a, Cr03, Cos04, Co07, To06b, Ji10a, St19}.
It is known that for a $d$-dimensional many-body system, one- and two-body correlation functions are insufficient to uniquely determine the higher-body correlation functions $g_3,g_4,...$ \cite{To06b}.
That is, systems with identical $g_2(r)$ for all $r$, or indistinguishable pair functions within numerical noise in simulations and experiments, can possess distinctly different higher-body correlation functions.
Such structural degeneracy indicates that there exist ambiguity in solutions to inverse statistical mechanical problems, in which one aims to infer inter-particle interactions from pair statistics alone \cite{Wa20}. 
The degeneracy of pair functions has been exactly shown for certain finite systems with a small number of particles $N\leq 30$ in two and three dimensions under rigid boundary conditions \cite{Ji10a, St19}.
Such degeneracy has also been numerically identified for disordered equilibrium and nonequilibrium systems in the thermodynamic limit \cite{Wa23}.
On the other hand, the absence of ground-state degeneracy, i.e., the uniqueness of the global energy minimum, can be rigorously proved for certain point configurations on a sphere under the action of isotropic pair potentials \cite{Co07}.
However, as noted above, structural degeneracy of pair functions for non-Bravais crystals have not been studied for $d \leq 3$. 

In this work, we ask the question, What is $n_{\text{min}}(d)$, the minimum value of $n$ for inequivalent crystals with the same theta series or, equivalently, radial distribution function, in space dimension $d$? 
Two crystals are said to be {\it equivalent} if one can be transformed into the other by isometric symmetry operations, including translations, rotations, reflections, and inversion with respect to a center.
Specifically, to tackle this question, we develop a precise numerical algorithm that searches for inequivalent isospectral crystals with a multi-particle basis in $d = 1, 2, 3$ dimensions.
We characterize the symmetry (i.e., crystal systems and space groups) of any identified isospectral crystals.
Furthermore, to investigate how isospectral crystals can arise in physical many-body systems, we use inverse statistical mechanical methods to search for interparticle interactions that yield degenerate  ground-state manifolds.

Our algorithm identifies inequivalent isospectral 4-, 3- and 2-particle bases in one, two and three dimensions, respectively. 
In one dimension, no isospectral crystal with $n\leq 3$ are identified by our algorithm.
That is, the minimum $n$ for inequivalent isospectral crystals is $n_{\text{min}}(1) = 4$.
We provide a straightforward proof for why $n_{\text{min}}(1)$ must be greater than 3.
This demonstrates the efficiency and reliability of the algorithm.
In two dimensions, no isospectral crystals with $n\leq 2$ are identified, which strongly suggests that $n_{\text{min}}(2) = 3$.
The algorithm also does not identify 3D Bravais lattice, consistent with the result proved by Schiemann \cite{Sc90b}, i.e., $n_{\text{min}}(3) = 2$.
Taken together, we conjecture that $n_{\text{min}}(d) = 4, 3, 2$ for $d = 1, 2, 3$, respectively.
Furthermore, we show that certain 2D isospectral crystals with 3-particle bases can possess high crystal geometry, including configurations with reflection planes, rotation axes, and inversion centers.
We also introduce a theorem that enables one to rigorously conclude that many isospectral crystals identified via our numerical algorithm indeed possess identical theta series up to infinite pair distances.

Targeting inequivalent isospectral crystals under the action of isotropic pair potentials enables us to study systems that possess degenerate crystalline ground states.
To study the degenerate ground-state manifolds associated with isotropic pair potentials, we use an inverse technique developed in Ref. \cite{Re06a}, which determines pair interactions that yield ground states with targeted radial distribution functions.
Via this method, we find an isotropic pair potential whose low-temperature configurations in two dimensions obtained via simulated annealing can lead to both of two isospectral crystal structures with 3-particle bases, one has cmm symmetry and contains linear chains of particles, while the other has p3m1 symmetry and contains equilateral triangular clusters.
We show that with sufficiently slow cooling rates, the two structures occur with approximately equal probabilities, but the structure with linear chains occurs more frequently with faster cooling. 

We begin by providing basic definitions and background in Sec. \ref{sec:def}.
Section \ref{sec:alg} provides a description our algorithm to search for inequivalent isospectral crystals.
Section \ref{sec:res} presents results for isospectral crystals for in one, two and three dimensions. 
Section \ref{sec:ground} describes the results for a many-body system under a pair potential that can lead to degenerate crystalline ground states, identified via the aforementioned inverse statistical-mechanical technique \cite{Re06a}.
We provide concluding remarks in Sec. \ref{sec:conc}. 

\section{Definitions and Preliminaries}
\label{sec:def}
In this section, we introduce some preliminary concepts and background that are crucial in our study of isospectral crystals with $n$-particle bases.

\subsection{Lattices, Crystals and Theta Series}
\label{sec:lattice}
A many-body system in $\mathbb{R}^d$ is completely statistically characterized by the $n$-particle probability density functions $\rho_n(\mathbf{r}_1,...,\mathbf{r}_n)$ for all $n\geq 1$ \cite{Ha86}.
In the case of statistically homogeneous systems, one has $\rho_1(\mathbf{r}_1)=\rho$ and $\rho_2(\mathbf{r}_1,\mathbf{r}_2)=\rho^2 g_2(\mathbf{r})$, where $\rho$ is the number density in the thermodynamic limit, $g_2(\mathbf{r})$ is the pair correlation function, and $\mathbf{r}=\mathbf{r}_2-\mathbf{r}_1$.
The \textit{radial distribution function} $g_2(r)$ is the angular-averaged pair correlation function, and it is defined such that $\rho s_1(r)g_2(r)dr$ gives the conditional probability of finding a particle center in the spherical shell of volume $s_1(r)dr$, given that there is another particle at the origin, where 
\begin{equation}
    s_1(r) = 2\pi^{d/2}r^{d-1}/\Gamma(d/2)
\end{equation}
is the surface area of a $d$-dimension sphere of radius $r$ \cite{To02a}.
Note that the functions $g_n(\mathbf{r}_1,...,\mathbf{r}_n) = \rho_n(\mathbf{r}_1,...,\mathbf{r}_n)/\rho^n$ apply to both disordered and ordered systems.

Beyond the pair correlation functions, we also study the three- and higher-body statistics $g_3, g_4, \dots$ corresponding to distributions of triangles, tetrahedra, etc., formed by points in a many-body system. 
We are particularly interested in the distribution of bond angles $\theta$ between pair displacement vectors of lengths on the order of $\rho^{-1/d}$. 
Thus, we express $g_3$ in terms of $\theta$, i.e.,
\begin{equation}
    g_3(r_1,r_2,\theta)= \frac{\rho_3\left(r_1,r_2,\sqrt{r_1^2+r_2^2-2ab\cos(\theta)}\right)}{\rho^3},
    \label{g3}
\end{equation}
where $\rho_3(r_1,r_2,r_3)$ is probability density of finding three particles that form a triangle with side lengths $r_1,r_2$ and $r_3$.

We recall from Sec. \ref{sec:intro} that a (Bravais) \textit{lattice} is a subgroup consisting of the integer linear combinations of \textit{basis vectors} $\mathbf{B} = \{\mathbf{a}_1, \dots, \mathbf{a}_d\}$ in $\mathbb{R}^d$.
In a lattice $\Lambda$, $\mathbb{R}^d$ can be divided into identical \textit{fundamental cells} (where $F$ denotes a fundamental cell), each of which containing just one point in $\Lambda$, whose position is given by the vector
\begin{equation}
    \mathcal{L} = n_1\mathbf{a}_1 + \dots n_d\mathbf{a}_d,
\end{equation}
where $n_i$ span all integers for $1 \leq i \leq d$.

A \textit{crystal} with an $n$-particle basis is a periodic point pattern obtained by placing a fixed configuration of $n \geq 1$ points within one fundamental cell $F$ of a lattice $\Lambda$, which is then periodically replicated. 
Note that a lattice is a crystal with $n = 1$.
A crystal $C$ is completely characterized by its underlying lattice $\Lambda_C$ and the position vectors of particles in a fundamental cell, denoted as $\mathbf{P}_C = \{\mathbf{p}_1, \dots, \mathbf{p}_n\}$, where
\begin{equation}
    \mathbf{p}_j = \nu_{j,1}\mathbf{a}_1 + \dots + \nu_{j,n}\mathbf{a}_d, \quad 0 \leq \nu_{j, i} < 1,
    \label{pj}
\end{equation}
for all integers $1 \leq i \leq d$ and $1 \leq j \leq n$.
Any position vector of point $\mathbf{p}\in C$ can be uniquely decomposed as $\mathbf{p} = \mathcal{L} + \mathbf{p}_j$, where $\mathcal{L} \in \Lambda_C$ and $\mathbf{p}_j \in \mathbf{P}_C$.
In this work, we assume that $\mathbf{p}_1 = \mathbf{0}$ in any crystal $C$, i.e., there exist particles at all points in $\Lambda_C$.

The \textit{theta series} for a crystal $C$ with respect to a point $\mathbf{p}\in C$ is defined as \cite{Co93}
\begin{equation}
    \Theta_{C}(q, \mathbf{p}) = \sum_{\mathbf{p'} \in C\setminus
    \{\mathbf{p}\}} q^{|\mathbf{p}' - \mathbf{p}|^2} = 1 + \sum_{m=1}^\infty b_m q^{X_m}
    \label{theta_def}
\end{equation}
where $q$ is a complex variable, $X_m$ is the squared norm of the vector to the $m$th nearest point from $\mathbf{p}$, and $b_m$ is the number of vectors of squared norm $X_m$.
The theta series encodes radial coordination information whereby $b_m$ is the number of particles at the squared norm $X_m$ from the particle at $\mathbf{p} \in C$ and exactly determines the radial distribution function $g_2(r)$  \footnote{For general crystals in which each particle does not have exactly the same radial coordination structure, the theta series (\ref{theta_def}) should be interpreted in an average sense, i.e., an average over each particle type. The corresponding radial distribution function $g_2(r)$ is obtained by averaging over each particle type $\mathbf{p}_j \in \mathbf{P}_C$.}.
Note that the partial sum up to $m=M$ provides the total number of particles in the crystal within distance $\sqrt{X_M}$ from $\mathbf{p}$ when $q=1$, i.e., the cumulative coordination number.
In the Appendix, we provide the first several terms of the theta series for some well-known 2D Bravais lattices ($n = 1$) and non-Bravais crystals ($n \geq 2$).
Evidently, for a lattice $\Lambda$, the coefficients of the theta series can be alternatively expressed as a multiset of pair distances:
\begin{equation}
    \Theta_\Lambda = \{(\sqrt{X_m}, b_m): m\in \mathbb{Z}^+\} = \{(r = |\mathcal{L}|, m_r): \mathcal{L}\in\Lambda\},
    \label{thetalambda}
\end{equation}
where $m_r$ is the multiplicity of the vector norm $r = |\mathcal{L}|$.
In this work, we generalize Eq. (\ref{thetalambda}) to a crystal $C$, defined as the multiset of all pair distances in $C$ such that one of the two particles is in $\mathbf{P}_C$; i.e.,
\begin{equation}
\begin{split}
   \Theta_C = \Theta_{\Lambda_C} \cup
   \{&(r = |\mathbf{p}_j - \mathbf{p}_{j'} + \mathcal{L}|, m_r):\\
   &\mathbf{p}_j, \mathbf{p}_{j'}\in \mathbf{P}_C, j\neq j', \mathcal{L}\in \Lambda_C\},
    \label{thetaC} 
\end{split}
\end{equation}
where, as above, $m_r$ is the multiplicity of the pair distance $r = |\mathbf{p}_j - \mathbf{p}_{j'} + \mathcal{L}|$.
We call two crystals $C_1, C_2$ \textit{isospectral} if they have identical theta series, or equivalently, if their corresponding $g_2(r)$ agree for all $r$.
Importantly, for two crystals to be isospectral, both the sets of pair distances and the multiplicities of each pair distance must match.

Two crystals $C_1, C_2$ are said to be \textit{equivalent} if one can be transformed to another via an \textit{isometry}, i.e., $C_1 = T(C_2)$ for some $T\in E(d)$, where $E(d)$ is the $d$-dimensional \textit{Euclidean group}.
The elements of $E(d)$ are transformations on $\mathbb{R}^d$ that preserve the Euclidean distance between any two points, including all translations, rotations, reflections, inversions with respect to a center, as well as all arbitrary finite compositions of them. 
Note that by this definition, enantiomeric chiral crystals are also considered equivalent as they are related by a reflection.
Equivalent crystals are trivially isospectral because the pair distances are preserved by the associated isometry.
In this work, we are interested in determining $n_{\text{min}}(d)$, the minimum value of $n$ for inequivalent isospectral crystals in $\mathbb{R}^d$. 

\subsection{Distance Metrics}
To study the isospectrality of crystals, it is useful to define a ``distance'' metric between the radial distribution functions for two crystals $C_1$ and $C_2$ at the same number density with a common underlying fundamental cell.
A metric that has been fruitfully employed in our previous works \cite{Wa20, To22} is given by the following $L_2$-norm distance between $g_2(r)$ for the two crystals: 
\begin{equation}
    D_{g_2}(C_1, C_2) = \rho\int_0^R [g_2^{(1)}(r) - g_2^{(2)}(r)]^2 s_1(r) d r,
    \label{dg2}
\end{equation}
where $g_2^{(i)}(r)$ is the radial distribution function of crystal $C_i$, and $R$ is an upper-cutoff pair distance, set to be twice the longest diagonal of the fundamental cell.
This choice of $R$ imposes the constraint that if $D_{g_2}(C_1, C_2) = 0$, then $C_1$ and $C_2$ must match at at least $4^d n^2$ pair distances, which well exceeds the number of degrees of freedom $n_F$ [see Eq. (\ref{dof})] for $d$-dimensional $n$-particle bases considered in this work.
Thus, it is reasonable to assume that crystals with matching $g_2(r)$ on $(0, R)$ also match at pair distances beyond this range.

Because numerical $g_2(r)$ functions are computed form binned histograms of pair distances, crystals with $D_{g_2}(C_1, C_2) = 0$ may not be exactly isospectral, because the pair distances in each crystal can vary within the range of a bin size and still yield the same binned $g_2(r)$.
Thus, it is useful to define another metric $D_\Theta(C_1, C_2)$ that measures the ``distance'' between the theta series of $C_1$ and $C_2$ by involving the exact pair distances:
\begin{equation}
    D_{\Theta}(C_1, C_2) = \sum_{j = 1}^M \left(r_j^{(1)} - r_j^{(2)}\right)^2,
    \label{Dtheta}
\end{equation}
where $r_j^{(i)}$ is the shortest $j$-th smallest pair distance in crystal $C_i$, $M$ is a positive integer, set to be $5^d n^2$ in this work. 
This value of $M$ is chosen because we intend to compare the pair distances up to twice the length of the longest diagonal of the fundamental cell. 
We thus consider the region consisting of $4^d=4\times4\times \cdots \times 4$ fundamental cells centered at the origin.
The pair distances between particles in this region is a subset of $\Theta_C$ given by
\begin{equation}
\begin{split}
    \theta_C &= \{(r = |n_1\mathbf{a}_1 + n_d\mathbf{a}_d|, m_r): n_1, \dots, n_d \in \{0, \pm1, \pm2\}\}\\
    &\cup\{(r = |\mathbf{p}_j-\mathbf{p}_j'+n_1\mathbf{a}_1 + n_d\mathbf{a}_d|, m_r): \\ 
    &j, j' = 1,\dots, n; j \neq j'; n_1, \dots, n_d \in \{0, \pm 1, \pm2\}\}.
\end{split}
\end{equation}
which contains $5^d(1 + n(n - 1)) \leq 5^dn^2$ pair distances. 
In all cases we study, we find that if two crystals match their pair distances within twice the longest diagonal of the fundamental cell, their pair distances continue to match at larger $r$.

To measure the degree to which two crystal configurations are geometrically different up to isometry, we define a geometric distance metric $\xi$ that vanishes if and only if $C_1$ and $C_2$ are equivalent:
\begin{equation}
    \xi(C_1, C_2) = \rho^{1/d} \min_{T\in E(d)}d_{\text{Ch}}(T(C_1), C_2).
    \label{xi}
\end{equation}
For crystals $C_1, C_2$ with respective fundamental cells $F_1, F_2$ and underlying lattices $\Lambda_{C_1}, \Lambda_{C_2}$, generated by the corresponding sets of basis vectors $\mathbf{B}_{C_1}$ and $\mathbf{B}_{C_2}$, we define $d_{\text{Ch}}(C_1, C_2)$ analogously to the Chamfer matching distance between two finite sets of points \cite{Va00}: 
\begin{equation}
\begin{split}
    d_{\text{Ch}}(C_1, C_2) &= \sum_{\mathbf{a}\in \mathbf{P}_{C_1}}d_{\mathbf{a}C_2} + \sum_{\mathbf{b}\in \mathbf{P}_{C_2}}d_{\mathbf{b}C_1} \\
    &+\sum_{\mathbf{p}\in \mathbf{B}_{C_1}}d_{\mathbf{p}\Lambda_{C_2}} + \sum_{\mathbf{q}\in \mathbf{B}_{C_2}}d_{\mathbf{q}\Lambda_{C_1}}, 
    \label{chamfer}
\end{split}
\end{equation}
where
\begin{subequations}
\begin{equation}
    d_{\mathbf{a}C_2} = \min_{\mathbf{b}\in \mathbf{P}_{C_2}, \mathcal{L}_2\in\Lambda_{C_2}}|\mathbf{a} - \mathbf{b} + \mathcal{L}_2|,
    \label{dac2}
\end{equation}
\begin{equation}
    d_{\mathbf{b}C_1} = \min_{\mathbf{a}\in \mathbf{P}_{C_1}, \mathcal{L}_1\in\Lambda_{C_1}}|\mathbf{b} - \mathbf{a} + \mathcal{L}_1|,
    \label{dbc1}
\end{equation}
\begin{equation}
    d_{\mathbf{p}\Lambda_{C_2}} = \min_{ \mathcal{L}_2\in\Lambda_{C_2}}|\mathbf{p} - \mathcal{L}_2|,
\end{equation}
\begin{equation}
    d_{\mathbf{q}\Lambda_{C_1}} = \min_{ \mathcal{L}_1\in\Lambda_{C_1}}|\mathbf{q} - \mathcal{L}_1|.
\end{equation}
\end{subequations}
The first two terms of Eq. (\ref{chamfer}) measure the degree to which the positions of points in the interior of fundamental cells $F_1$ and $F_2$ associated with crystals $C_1$ and $C_2$ match each other.
For example, $d_{\mathbf{a}C_2}$ (\ref{dac2}) represents the minimum distance from $\mathbf{a}\in \mathbf{P}_{C_1}$ to any of $\mathbf{b}$’s images in the crystal $C_2$. 
On the other hand, the third and fourth terms of Eq. (\ref{chamfer}) measure the degree to which $F_1$ and $F_2$ align with each other, i.e., the sum of these terms vanish if and only if $F_1$ and $F_2$ are identical regions in $\mathbb{R}^d$. 

\section{Algorithm to Search for Inequivalent Isospectral Crystals}
\label{sec:alg}
Here, we describe our precise algorithm that enables us to numerically determine $n_{\text{min}}(d)$ for $d = 1, 2, 3$.
Our algorithm is designed to tackle the isospectrality problem of crystals as a problem of multi-objective optimization.

The goal of our algorithm is to search for pairs of inequivalent isospectral crystals $C_1$, $C_2$ with $n$-particle bases.
For simplicity, we set the number density $\rho = 1$ for all crystals and consider only cases such that $C_1$ and $C_2$ have identical basis vectors associated with the fundamental cells, i.e., the two crystals differ only in the position vectors of the particles in the interior of the fundamental cell $\mathbf{p}_2^{(i)}, \dots, \mathbf{p}_n^{(i)}$, where the superscript $(i)$ refers to crystal $C_i$.
In general, we allow the fundamental cells to deform, such that the basis vectors vary while preserving the volume of the fundamental cell $|F|=n$. Thus, the optimization parameters are the basis vectors $\mathbf{a}_i, \dots, \mathbf{a}_d$, as well as the position vectors of the interior points for both crystals $\mathbf{p}_2^{(i)}, \dots, \mathbf{p}_n^{(i)}$ for $(i) = (1), (2)$.
The total number of degrees of freedom is given by 
\begin{equation}
    n_F = d^2 + 2d(n - 1) - \frac{d (d - 1)}{2} - 1,
    \label{dof}
\end{equation}
where the first term corresponds to the scalar components of the $d$ basis vectors in $\mathbb{R}^d$, the second term corresponds to the scalar components of $2(n-1)$ interior particles, the third term comes from the rotation degrees of freedom of the fundamental cell, and the final term comes from the constraint that $\rho = 1$.
In the Appendix, we describe our specific methods to eliminate the $d(d-1)/2 + 1$ degrees of freedom due to rotations and scaling, thereby transforming the vectors $\mathbf{a}_1, \dots, \mathbf{a}_d$ and $\mathbf{p}_2^{(i)}, \dots, \mathbf{p}_n^{(i)}$ into $n_F$ free scalar parameters subject to optimization.
To determine $n_{\text{min}}(d)$, we perform the ``full'' optimization with deformable fundamental cells.
However, we also perform additional optimization with fixed fundamental cells to search for isospectral crystals with high symmetry, such as 2D crystals with square, hexagonal and rectangular fundamental cells, as well as 3D crystals with cubic and tetragonal fundamental cells.
Specifically, in fixed-cell simulations, we fix the basis vectors $\mathbf{a}_1, \dots, \mathbf{a}_d$ and vary only the positions of the interior particles for both crystals, thereby obtaining pairs of isospectral crystals with different symmetries.

To enforce isospectrality on $C_1$ and $C_2$, one must attempt to decrease the metric $D_{g_2}(C_1, C_2)$ (\ref{dg2}) that measures the distance between the the radial distribution functions of the two crystals.
The integral in (\ref{dg2}) is computed as a Riemann sum of the integrand and hence $g_2^{(i)}(r)$ is computed from a histogram of pair distances with bin size $0.1\rho^{-1/d}$.

As expected, we find that simulated annealing procedures whose sole objective is to minimize $D_{g_2}(C_1, C_2)$ invariably yield equivalent crystals, which are trivially isospectral.
This difficulty to find inequivalent isospectral cases is due to the fact that equivalent crystals occupy a much larger region in the parameter space than inequivalent isospectral crystals.
To tackle this challenge, we require that while $D_{g_2}(C_1, C_2)$ decreases, the algorithm must attempt to increase the geometric distance metric $\xi(C_1, C_2)$ between the two crystals (\ref{xi}).

Because the group $E(d)$ contains infinitely many translations and rotations, it is computationally expensive to compute Eq. (\ref{xi}) by minimizing over all $T\in E(d)$.
Thus, we consider only the set $E'(d) \subseteq E(d)$ consisting of the transformations $T$ such that $T(C_1)$ and $C_2$ have identical partition of $\mathbb{R}^d$ by the fundamental cells and that $\mathbf{p} = \mathbf{p}'$ for some $\mathbf{p}\in T(C_1), \mathbf{p}' \in C_2$.
Qualitatively, the elements in $E'(d)$ consist of transformations $T$ that yield small values of the Chamfer matching distance $d_{\text{Ch}}(T(C_1), C_2)$, because $T(C_1)$ and $C_2$ match exactly at at least one point per fundamental cell.
The elements of $E'(d)$ are exactly the following transformations:
\begin{itemize}
    \item The $n^2$ translations by vectors $\mathbf{a}-\mathbf{b}$, where $\mathbf{a}\in \mathbf{P}_{C_1}$ and $\mathbf{b}\in \mathbf{P}_{C_2}$. 
    We denote by $\mathbf{T}_1$ the set of such transformations.
    \item The transformations in the point group $\mathbf{T}_2$ associated with the symmetry of the fundamental cell. 
    For example, for a 2D rectangular fundamental cell, the corresponding transformations are the inversion with respect to the center of the fundamental cell, rotation by $\pi$, and the two reflections along the lines of symmetry of the fundamental cell.
    \item The compositions $T_2\circ T_1$, where $T_1 \in \mathbf{T}_1$ and $T_2\in \mathbf{T}_2$.
\end{itemize}
For any given fixed fundamental cell, $E'(d)$ is a finite set.
However, when a deformable fundamental cell is used, the cell can attain shapes that closely match the symmetries of any $d$-dimensional point group during the optimization procedure.
Thus, one must include in $\mathbf{T}_2$ the symmetry operations in all point groups in $\mathbb{R}^d$, which again contains an infinite number of rotations.
To optimize over rotations, we use the multi-start Broyden–Fletcher–Goldfarb–Shanno (BFGS) algorithm \cite{Liu89} starting from a grid of initial rotation angles on $[0, 2\pi)$ with $5^\circ$ intervals for 2D crystals, and 100 randomly generated 3D rotation matrices based on the Haar measure \cite{St80} for 3D crystals.
BFGS is an efficient quasi-Newtonian optimization technique to find deep local minima by incorporating both the gradient and an iteratively improved approximant of the Hessian \cite{Liu89}.
In our implementation of the algorithm, $\xi(C_1, C_2)$ is computed by replacing $E(d)$ in Eq. (\ref{xi}) with $E'(d)$. 
Note that with this simplification, $\xi(C_1, C_2)$ still satisfies the desired property that it vanishes if and only if $C_1$ and $C_2$ are equivalent.

To simultaneously decrease $D_{g_2}(C_1, C_2)$ and increase $\xi(C_1, C_2)$, we perform simulated annealing optimization on the following scalar objective function
\begin{equation}
    \Psi(C_1, C_2) = D_{g_2}(C_1, C_2) - c\exp[-\alpha D_{g_2}(C_1, C_2)]\xi(C_1, C_2),
    \label{Psi}
\end{equation}
where $c$ and $\alpha$ are positive parameters.
The form of the second term in $\Psi(C_1, C_2)$ is designed with the consideration that during the simulated annealing procedure, it is desirable to first decrease the difference in $g_2(r)$ of the two crystals at higher temperatures.
As $g_2(r)$'s of the crystals become closer, we give the metric $\xi(C_1,C_2)$ a larger weight to encourage the formation of geometrically different crystals. 
Thus, we include a weight factor $\exp[-\alpha D_{g_2}(C_1, C_2)]$ in the second term.

\begin{figure*}[htp]
    \centering
    \subfloat[]{\includegraphics[width = 60mm]{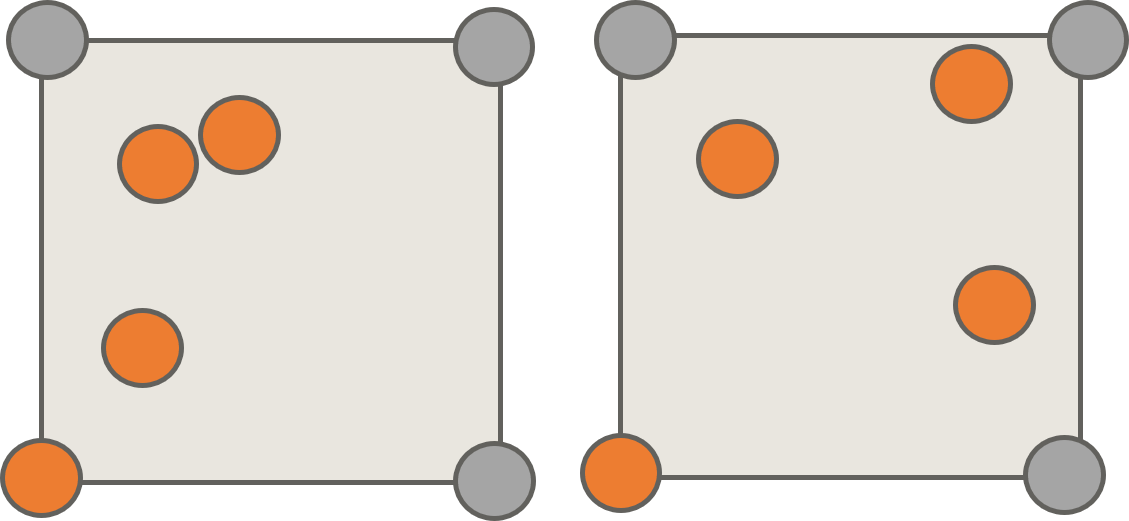}}
    \hspace{10mm}
    \subfloat[]{\includegraphics[width = 60mm]{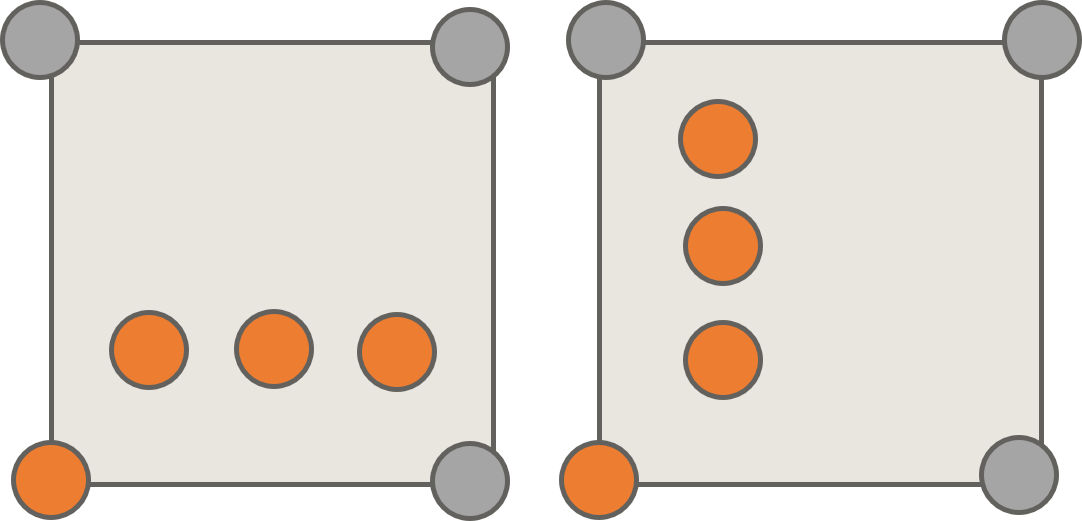}}
    
    \subfloat[]{\includegraphics[width = 60mm, trim={0 -1cm 0 0},clip]{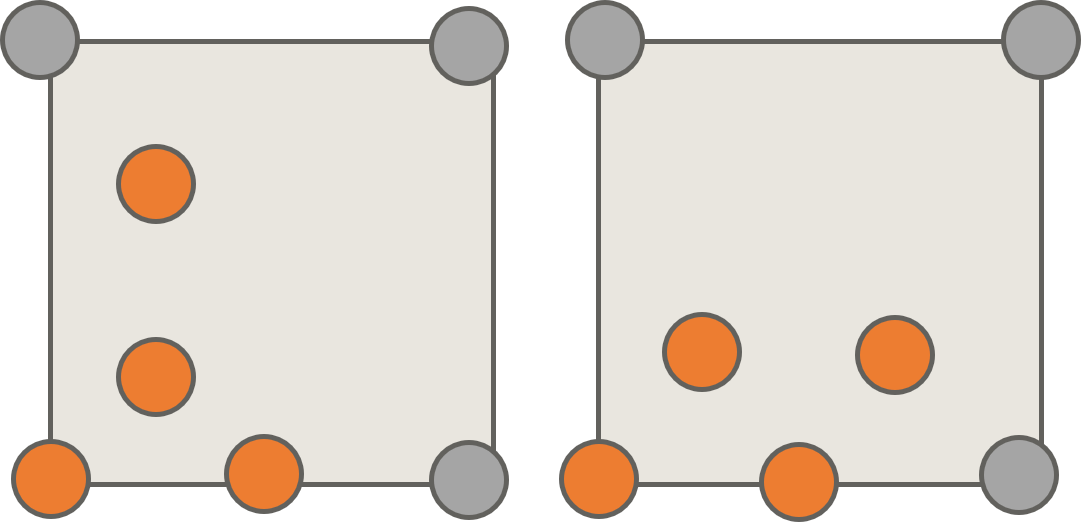}}
    \hspace{10mm}
    \subfloat[]{\includegraphics[width = 60mm]{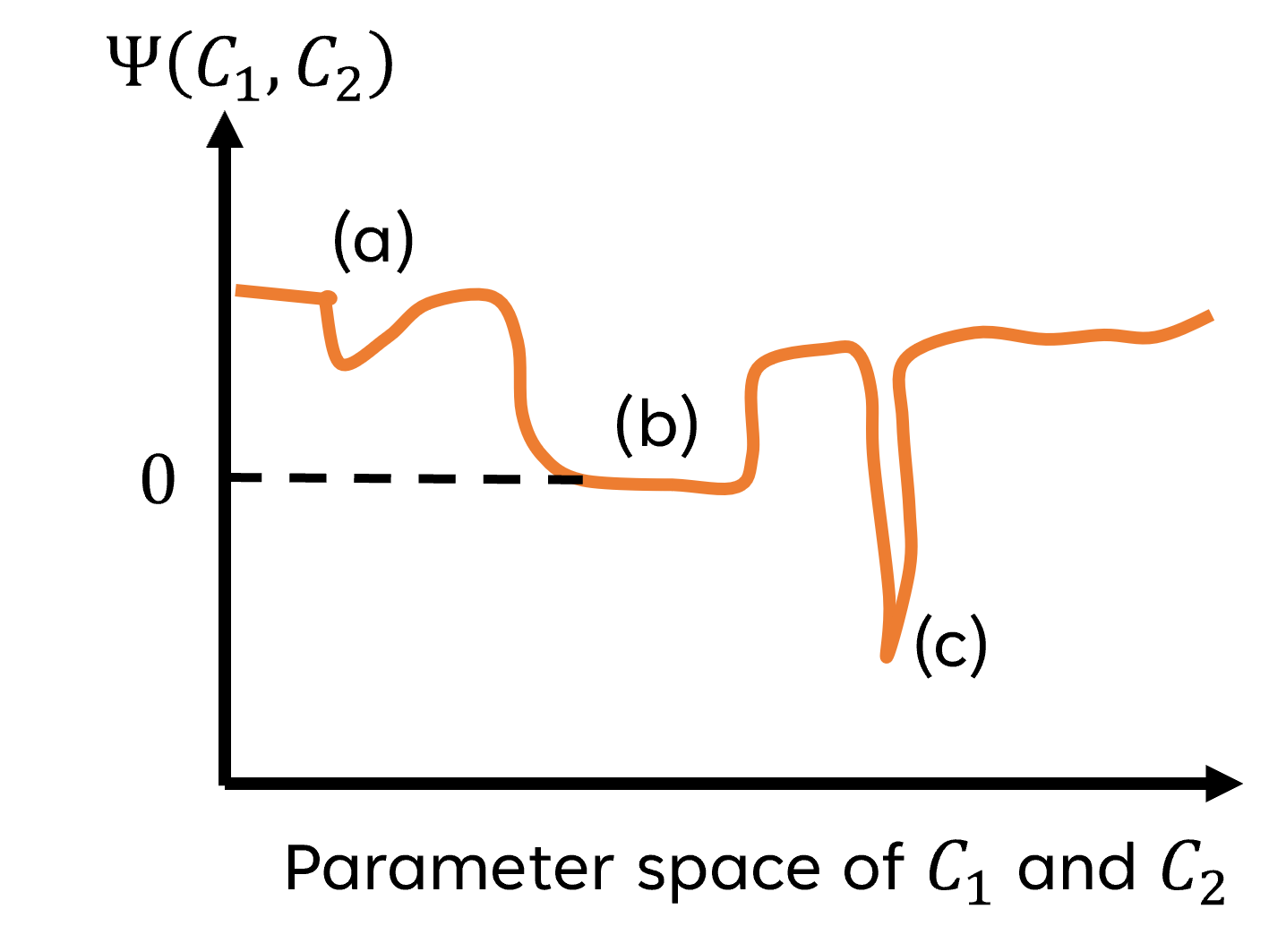}}
    \caption{Illustration showing the objective of the algorithm to search for pairs of inequivalent isospectral crystals in the case $d = 2, n = 4$.
    The orange particles are particles in the fundamental cell $\mathbf{p}_1, \dots \mathbf{p}_4$ [see Eq. (\ref{pj})], and the gray particles are images of $\mathbf{p}_1 = \mathbf{0}$.
    (a) A case with high $D_{g_2}$ and high $\xi$. 
    (b) A case with low $D_{g_2}$ and low $\xi$, which yields crystal structures that are approximately equivalent.
    (c) A case with low $D_{g_2}$ and high $\xi$, which yields crystals that are close to being inequivalent isospectral. 
    The optimization algorithm aims to search for cases (c) by minimizing the objective function $\Psi$ that involves both $D_{g_2}$ and $\xi$ and attains its deep local minima at cases of inequivalent isospectral crystals.
    (d) Schematic of the ``energy landscape'' of the optimization objective $\Psi$ as a function of the crystal parameters of $C_1$ and $C_2$. 
    The cases (a)--(c) are indicated as local minima in the energy landscape.}
    \label{fig:Algo}
\end{figure*}

Figure \ref{fig:Algo} schematically illustrates the idea behind the algorithm for the special case in $\mathbb{R}^2$ with $n = 4$.
The local minima of $\Psi$ are of three types, illustrated in Figs. \ref{fig:Algo}(a)--(c).
Type (a) minima have high $D_{g_2}$ ($\sim 10^1$) and high $\xi$ ($\sim 2n$), corresponding to crystals that are not isospectral and distinctly different in geometry.
Type (b) minima have low $D_{g_2}$ ($\sim 10^{-1}$) and low $\xi$ ($\sim 10^{-2}$) and correspond to nearly equivalent crystals.
Type (c) minima have low $D_{g_2}$ ($\sim 10^{-1}$) and high $\xi$ ($\sim 10^0$), corresponding to inequivalent isospectral crystals.
The ordering of values of $\Psi$ at these three types of local minima depends on $c$ and $\alpha$. 
For example, in the limit of large $c$ and small $\alpha$, type (a) minima attain lower $\Psi$ than type (c) minima.
To search for inequivalent isospectral crystals,  $c$ and $\alpha$ in (\ref{Psi}) must be tuned such that $\Psi$ attains its deep local minima at cases (c), as illustrated in Fig. \ref{fig:Algo}(d).
In practice, for crystals at unit density with deformable fundamental cells, we find that the choices $c = 30, \alpha = 0.1$ enable us to identify inequivalent isospectral crystals for $d = 1, 2, 3$.
With fixed fundamental cells, we use $c = 10, \alpha = 0.5$.
Different simulated annealing trajectories usually attain different deep local minima for $\Psi$, corresponding to different pairs of inequivalent isospectral crystals, if they exist for given $n$ and $d$.
Since our goal is simply to find  inequivalent isospectral crystals, we do not require the algorithm to find the true global minimum of $\Psi$.
The simulated annealing procedure terminates when a case is found such that $D_{g_2}(C_1, C_2) = 0$ and $\xi(C_1, C_2) > 0.5$.

As noted in Sec. \ref{sec:def}, the optimization procedure described above yields approximately, but not exactly, isospectral crystals, due to the finite bin size of numerical $g_2(r)$.
To generate crystals that are exactly isospectral on $[0, R]$ within the machine epsilon, we use the approximately isospectral crystals found via simulated annealing as initial inputs and refine the crystal parameters by minimizing the metric $D_\Theta(C_1, C_2)$ [Eq. (\ref{Dtheta})] via the BFGS algorithm \cite{Liu89}.
The refinement procedure proceeds until the stopping criterion $D_{\Theta}(C_1, C_2) < 10^{-8}$ is satisfied.
Note that, in this final refinement stage, the objective function (\ref{Dtheta}) does not involve $\xi$, because it is assumed that particle positions will only vary slightly (compared to the characteristic length scale $\rho^{-1/d}$) during the BFGS optimization procedure.
Because the initial configurations before refinement have $\xi\sim 1$, we assume that $\xi$ will remain on this order of magnitude during refinement, which is what we have observed for all cases of inequivalent isospectral crystals identified in Sec. \ref{sec:res}.

\section{Results for Isospectral Crystals for $d = 1, 2, 3$}
\label{sec:res}
In this section, we present results for the minimum multi-particle basis of inequivalent isospectral crystals., $n_{min}(d)$, in one, two and three dimensions identified via the algorithm in Sec. \ref{sec:alg}.
These results lead to our conjecture that $n_{\text{min}}(d) = 4, 3, 2$ for $d = 1, 2, 3$, respectively.
For $d = 1$, we provide a rigorous proof that $n_{\text{min}}(1) = 4$, demonstrating the accuracy of our algorithm.
We also introduce a theorem on isospectrality of crystals, which enable us to rigorously show that certain 2D inequivalent isospectral crystals with $n = 3$ identified via our numerical algorithm indeed possess identical $g_2(r)$ for all $r$, despite the fact that the algorithm considers only a finite range of pair distances.
For $d = 3$, our numerical result $n_{\text{min}}(3) = 2$ is consistent with the rigorous result \cite{Sc90b} that 3D inequivalent isospectral lattices to not exist, which further validates the precision and power of the algorithm.

\subsection{1D Cases}
Our algorithm identifies a pair of 1D inequivalent isospectral crystals with $n = 4$, whose configurations are shown in Fig. \ref{fig:1d}(a).
This is the only case of 1D inequivalent isospectral 4-particle bases identified by the algorithm.
Here, the shortest radial distance $r_1 = L/13$ and $L = 4\rho^{-1}$ is the lattice constant of the underlying Bravais lattice.
The two crystals are clearly inequivalent, because the segments of lengths $r_1$ and $2r_1$ are adjacent in $C_1$, but are separated by a segment of length $3r_1$ in $C_2$.
No inequivalent isospectral crystals with $n \leq 3$ are found, i.e., $n_{\text{min}}(1) = 4$ according to our numerical methods.

\begin{figure*}[htp]
    \centering
    \subfloat[]{\includegraphics[width = 100mm]{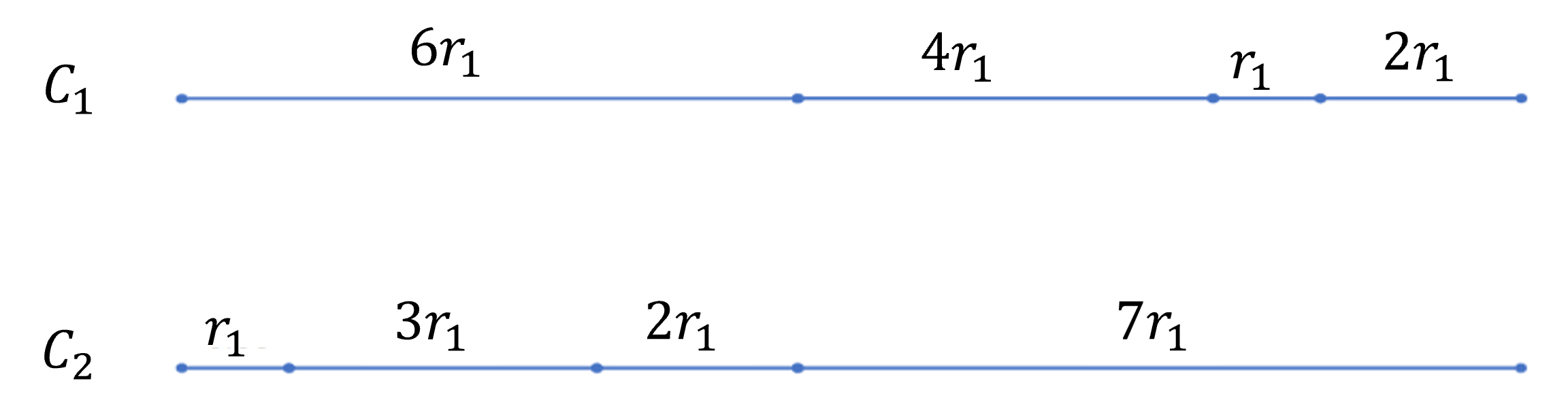}}
    
    \subfloat[]{\includegraphics[width = 160mm]{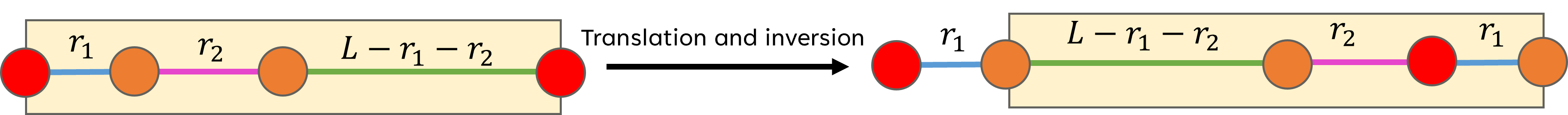}}
    \caption{(a) The case of 1D inequivalent isospectral crystals with 4-particle bases identified via our algorithm. 
    (b) Illustration showing that two 1D isospectral crystals with 3-particle bases are equivalent, as they are related by a translation and an inversion; see the configurations in the fundamental cells indicated by the yellow boxes.}
    \label{fig:1d}
\end{figure*}

We now prove rigorously that $n_{\text{min}}(1) = 4$.
To show the isospectrality of two 1D crystals $C_1, C_2$, one only needs to show that the multisets of radial distances that are smaller than the lattice constant 
\begin{equation}
    \theta_{C_i} = \{(r = |p_j - p_{j'}|, m_r): p_j, p_{j'} \in \mathbf{P}_{C_i}; j<j'\}
\end{equation}
are identical for both crystals.
To show that $\theta_{C_1} = \theta_{C_2}$ implies $\Theta_{C_i} = \Theta_{C_2}$, we observe that due to Eq. (\ref{thetaC}), all pair distances in $\Theta_{C_i}$ are of the form $||p_j - p_{j'}| + n_1L|$ for $|p_j - p_{j'}| \in \theta_{C_i}$ and integer $n_1$.
For the crystals with $n = 4$ shown in Fig. \ref{fig:1d}(a), one finds by enumerating the pair distances, counting multiplicity, that
\begin{equation}
    \theta_{C_1} =\theta_{C_2} = \{(xr_1, 1) : x = 1,\dots, 12\}.
\end{equation}
which proves that $C_1$ and $C_2$ are isospectral, i.e., $n_{\text{min}}(1) \leq 4$.

We now prove that 1D inequivalent isospectral crystals with $n\leq 3$ do not exist. 
For $n = 3$, given the underlying lattice constant $L = 3\rho^{-1}$ and the shortest two radial distances $r_1, r_2$, the fundamental cell of a 3-particle basis can be partitioned into line segments of lengths $r_1, r_2$, and $L - r_1 - r_2$. 
Because the fundamental cell is periodically repeated, the crystals created by all 6 perturbations of these 3 segments in the fundamental cell are related by isometries, which can be shown by considering each of the 6 perturbations.
An example of such a perturbation is given in Fig. \ref{fig:1d}(b).
Here, the three line segments in the fundamental cells delineated with red particles are ordered differently.
However, the two crystals are equivalent because they are related by the composition of a translation and an inversion, and this equivalence is evident from Fig. \ref{fig:1d}(b) when one considers the alternative choice of fundamental cells indicated with the yellow boxes.
One can similarly consider the other perturbations.
Thus, a 1D crystal with $n = 3$ is completely determined, up to isometry, by $r_1$ and $r_2$ in its theta series.
Similarly, for $n = 2$, the fundamental cell is partitioned into line segments of lengths $r_1$ and $L - r_1$. 
The two crystals created by the two perturbations of these line segments are clearly related by an inversion.
The case of $n = 1$ is trivial.
To summarize, the numerical result $n_{\text{min}}(1) = 4$ determined by  our algorithm exactly matches the result derived using rigorous methods, demonstrating the efficiency and reliability of the algorithm.

\subsection{2D Cases}
Our algorithm identifies 2D inequivalent isospectral crystals with $n = 3$, some examples of which are shown in Fig. \ref{fig:2D_full} and \ref{fig:2D}.
No 2D inequivalent isospectral crystals with $n \leq 2$ are found by the algorithm, i.e., $n_{\text{min}}(2) = 3$ according to our numerical methods.
We expect the 2D results via the algorithm to be highly reliable, given that the results of its 1D counterpart was proven to be exact, even though it is challenging to rigorously prove that 2D inequivalent isospectral 2-particle bases do not exist.

\begin{figure*}[htp]
    \centering
    \subfloat[]{\includegraphics[width = 70mm]{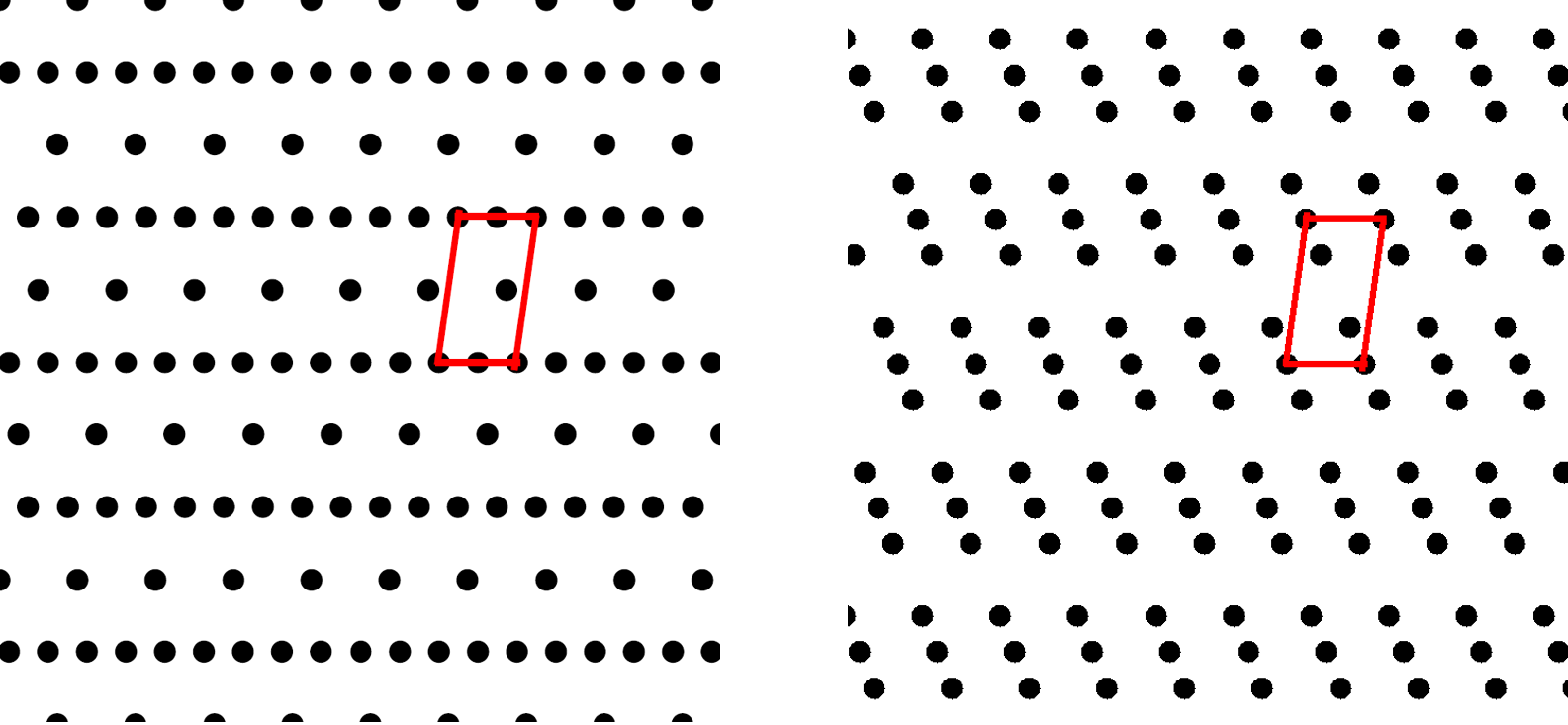}}
    \hspace{10mm}
    \subfloat[]{\includegraphics[width = 50mm]{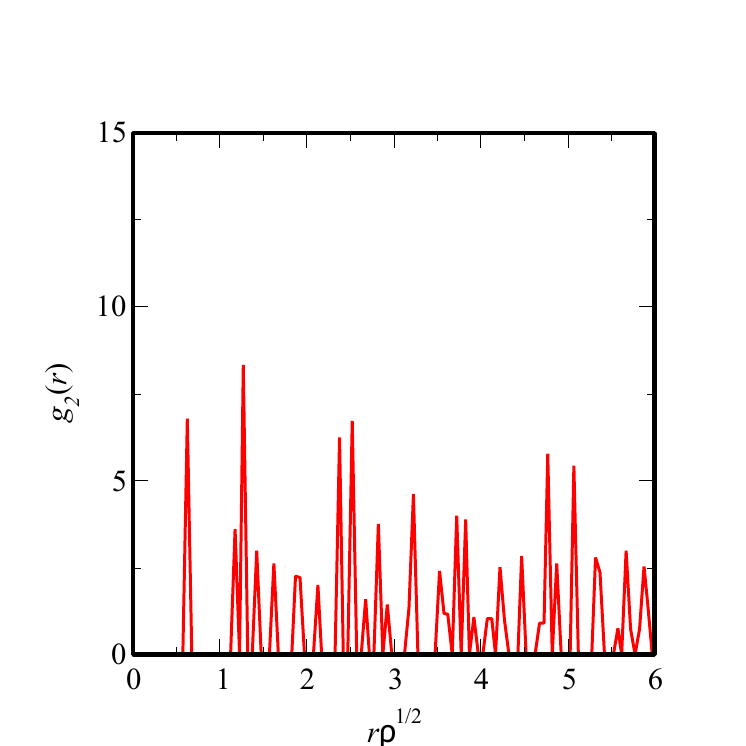}}
    \caption{(a) Configurations of the pair of 2D inequivalent isospectral crystals with 3-particle bases that attains the lowest value of $\Psi(C_1, C_2)$ (\ref{Psi}) among all isospectral 2D 3-particle bases identified via our algorithm using deformable fundamental cells. 
    The fundamental cells are oblique and are are indicated with parallelograms with red borders.
    (b) Radial distribution function $g_2(r)$ of both crystals in (a).}
    \label{fig:2D_full}
\end{figure*}

\begin{figure*}[htp]
    \centering
    \subfloat[]{\includegraphics[width = 70mm]{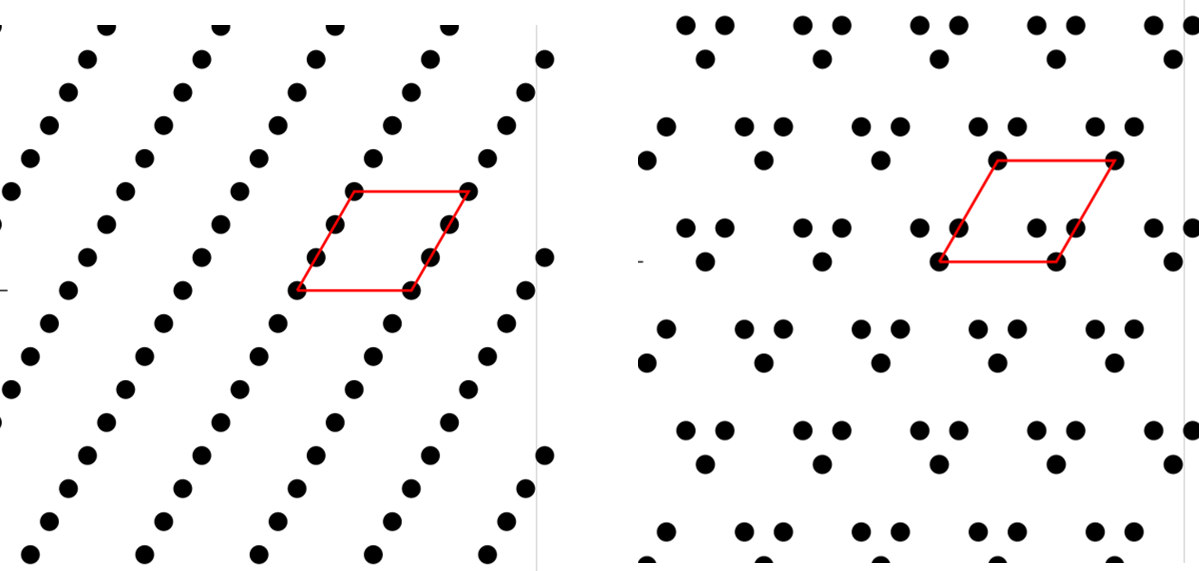}}
    \hspace{10mm}
    \subfloat[]{\includegraphics[width = 50mm]{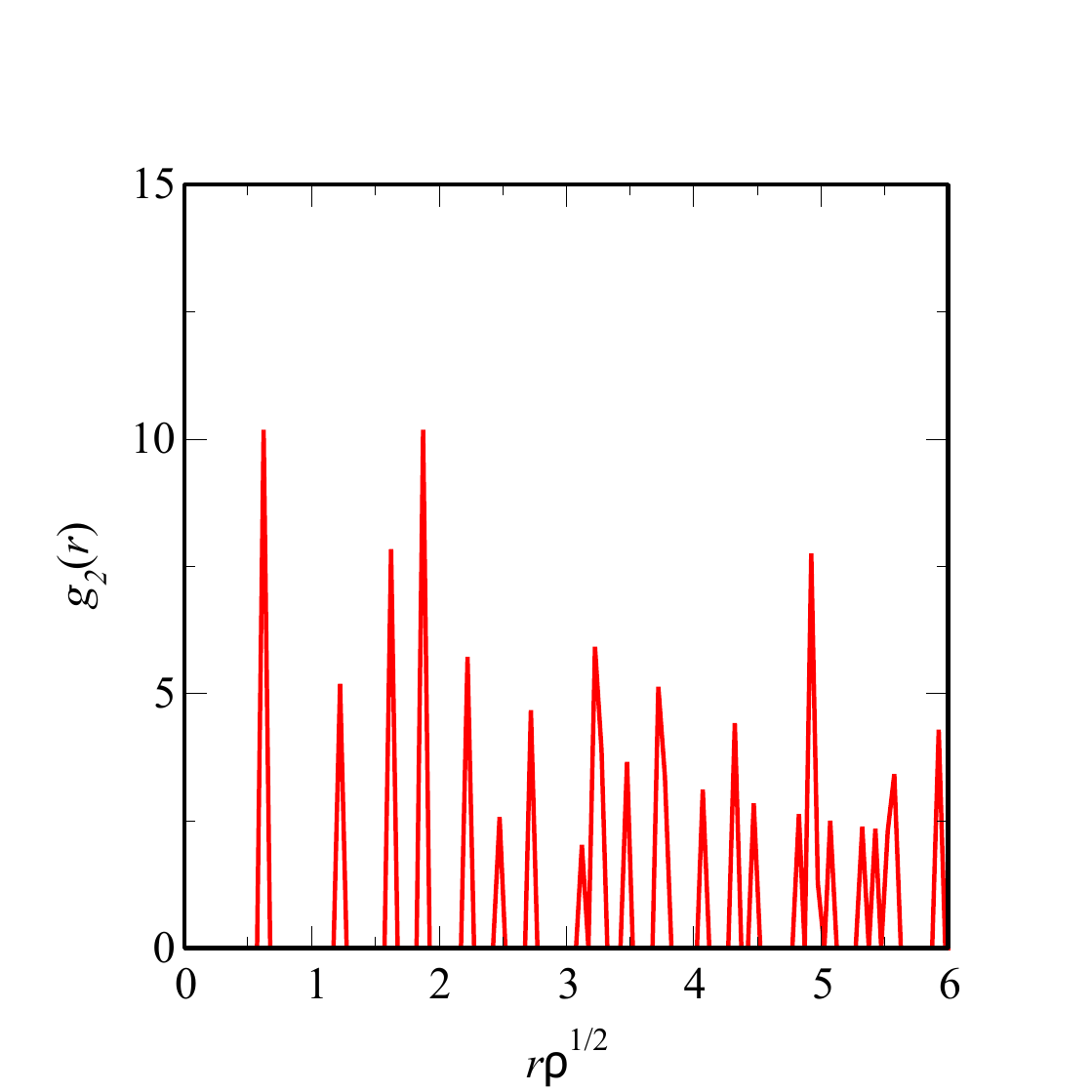}}
    
    \subfloat[]{\includegraphics[width = 70mm]{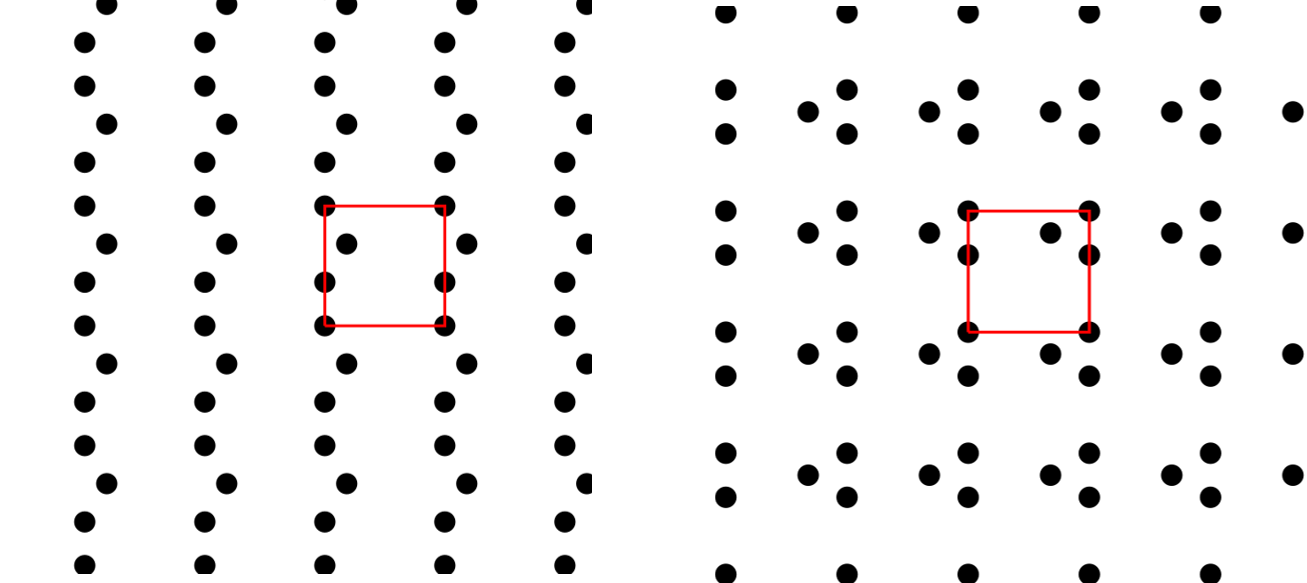}}
    \hspace{10mm}
    \subfloat[]{\includegraphics[width = 50mm]{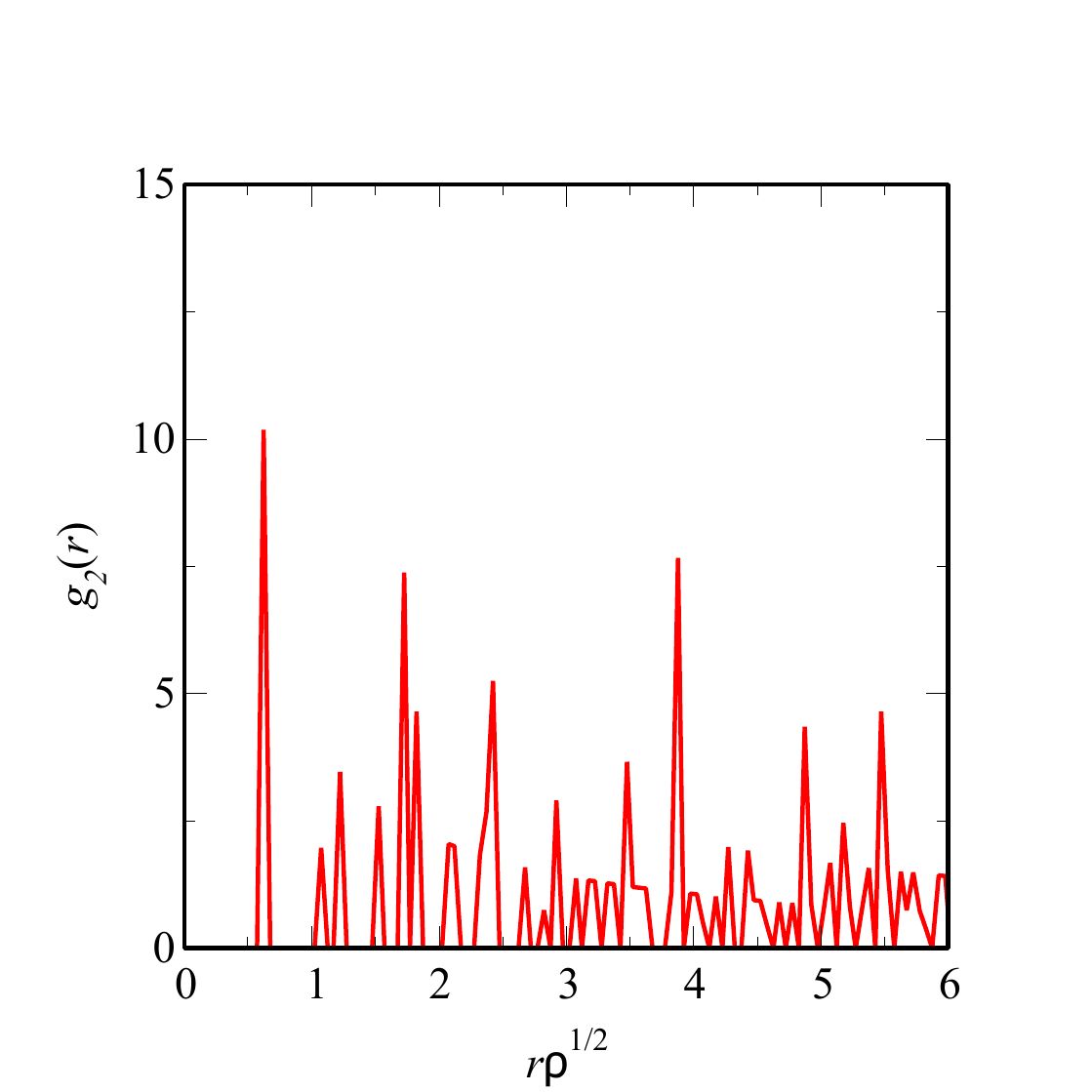}}
    
    \subfloat[]{\includegraphics[width = 70mm]{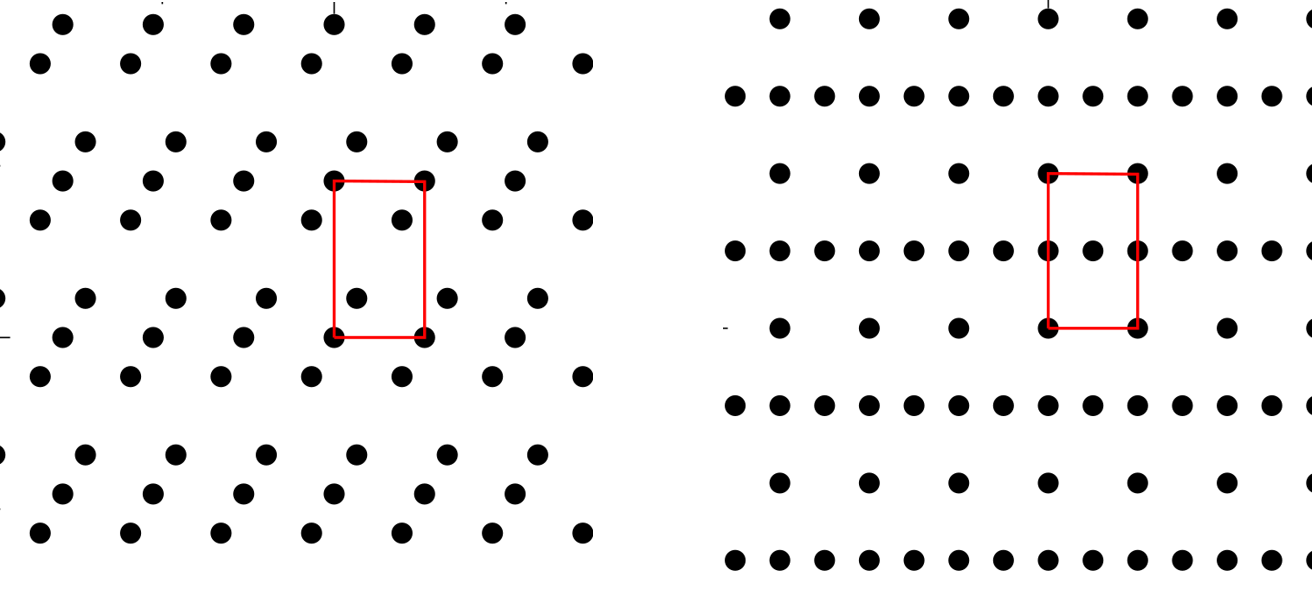}}
    \hspace{10mm}
    \subfloat[]{\includegraphics[width = 50mm]{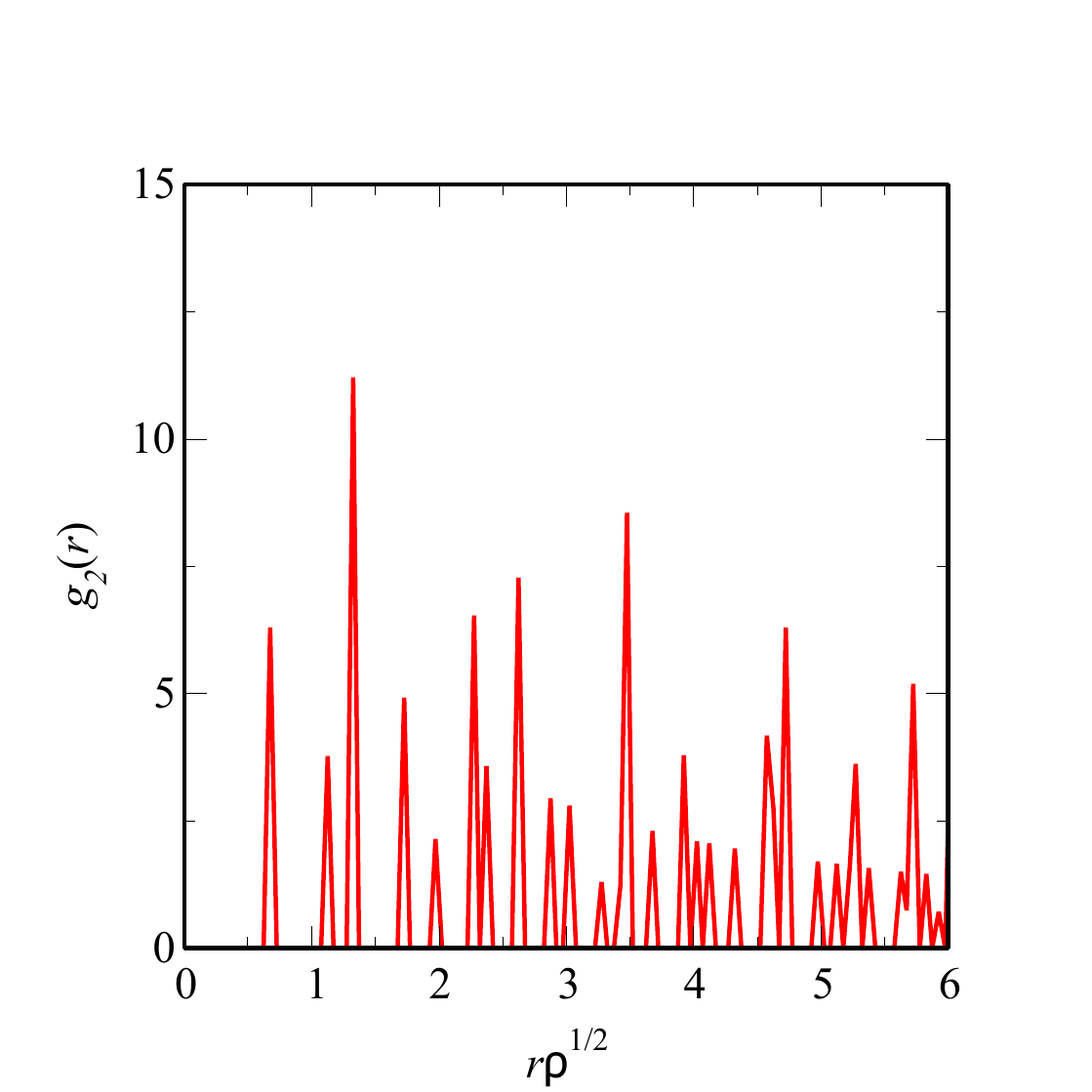}}
    \caption{Configurations of 2D inequivalent isospectral crystals with 3-particle bases identified via our algorithm using fixed fundamental cells that have more symmetry elements than simple oblique cells. 
    The fundamental cells are indicated with parallelograms with red borders.
    (a) A pair of isospectral crystals with hexagonal fundamental cells, with cmm and p3m1 symmetries, respectively.
    (b) Radial distribution function $g_2(r)$ of both crystals in (a).
    (c) A pair of isospectral crystals with square fundamental cells, both with pm symmetry.
    (d) $g_2(r)$ of both crystals in (c).
    (e) A pair of isospectral crystals with rhombic fundamental cells, with pgg and pmm symmetries, respectively.
    (f) $g_2(r)$ of both crystals in (e).
    All radial distribution functions are plotted with bin size $0.05\rho^{-1/2}$.
    }
    \label{fig:2D}
\end{figure*}

Figure \ref{fig:2D_full}(a) shows configurations of a pair of 2D inequivalent isospectral crystals with 3-particle bases, identified via our algorithm using deformable fundamental cells.
This pair attains the minimum value of our objective function $\Psi(C_1, C_2)$ (\ref{Psi}) among all isospectral 2D 3-particle bases found in this study.
The $g_2(r)$ for both crystals is shown in Fig. \ref{fig:2D_full}(b).
Because $D_{g_2}(C_1, C_2) = 0$ for isospectral crystals, the pair of crystals shown in Fig. \ref{fig:2D_full}(a) maximizes the geometric distance metric $\xi(C_1, C_2)$ (\ref{xi}).
The optimized parameters are presented in Table \ref{tab:crystalparams}, from which it is clear that the crystals have oblique fundamental cells.
The space group for both crystals is p2.
It is expected that isospectral crystals with large values of $\xi(C_1, C_2)$ have low symmetries of the fundamental cell, because the last two terms of $d_{\text{Ch}}(T(C_1), C_2)$ (\ref{chamfer}), which measure the alignment of the fundamental cells $T(F_1)$ and $F_2$, vanish for any symmetry operation $T$ consistent with the point group of the fundamental cell, thereby decreasing the value of $\xi(C_1, C_2)$ for highly symmetrical fundamental cells.

To study whether 2D isospectral crystals with 3-particle bases can attain symmetries other than the p2 group, we perform optimization using fixed fundamental cells with more symmetry elements than simple oblique cells.
We have found cases with hexagonal [Fig. \ref{fig:2D}(a)], square [\ref{fig:2D}(c)], rectangular [\ref{fig:2D}(e)] fundamental cells, and $g_2(r)$ for each pair of inequivalent isospectral crystals are given in Figs. \ref{fig:2D}(b), (d), (f), respectively.
All radial distribution functions are plotted with bin size $0.05\rho^{-1/d}$, and the finite widths and heights of the peaks in Figs. \ref{fig:2D}(b), (d), (f) are due to this binning of pair distances.
However, note that because the crystals parameters have been refined by minimizing $D_\Theta(C_1, C_2)$ (\ref{Dtheta}), the plotted $g_2(r)$ match for arbitrarily small bin sizes.
The isospectral 3-particle bases can possess high crystal symmetry, with symmetry elements such as 3-fold rotational axes, inversion centers, and reflection lines.
Specifically, the space groups for the crystals in Fig. \ref{fig:2D} are given by (a) cmm and p3m1, (c) both pm, and (e) pgg and pmm.
As expected, due to the additional constraints on fundamental cells, all of the high-symmetry cases shown in Fig. \ref{fig:2D} yield higher values of $\Psi(C_1, C_2)$ than from those cases shown in Fig. \ref{fig:2D_full} with oblique fundamental cells found via the full optimization.
Taken together, we have shown that there exist inequivalent isospectral 3-particle bases in all four 2D crystal systems.

Furthermore, Figs. \ref{fig:2D}(a) and (b) show that pairs of 2D inequivalent isospectral crystals often contain ``striped'' and ``clustered'' motifs, respectively. 
These distinct motifs occur due to the fact that our algorithm attempts to minimize the difference in $g_2(r)$ while maximizing the geometric distance $\xi(C_1, C_2)$.
Because our objective function (\ref{Psi}) is designed such that $D_{g_2}$ is reduced to small values in early stages of the simulated annealing procedure at relatively high temperatures, the algorithm must attempt to maximize the difference in three- and higher-body correlation functions at later stages with low temperatures.
Our previous works have shown that disordered systems with degenerate $g_2(r)$ can attain significant differences in $g_3$ at configurations of small triangles \cite{Wa20, Wa23}. 
Here, we show that such differences in $g_3$ also arise in crystal states.

Let $r_1 \approx r_2$ be the smallest two pair distances and consider the bond angle distribution $g_3(r_1, r_2, \theta)$ defined in Eq. (\ref{g3}). 
Striped motifs attain higher values of $g_3(r_1, r_2, \theta)$ at approximately linear 3-point configurations, where $\theta \approx \pi$, whereas clustered motifs attain higher values of $g_3(r_1, r_2, \theta)$ for equilateral-triangle-like configurations with $\theta \approx \pi/3$.
The striped and clustered motifs found in the inequivalent isospectral crystals  clearly indicates that their higher-order correlation functions are different.
For example, in the case of Fig. \ref{fig:2D}(a), we have $r_1 = r_2 = (\sqrt{2\sqrt{3}}/3)\rho^{-1/2}$.
The striped structure has a delta peak at $g_3(r_1, r_2, \pi)$ and vanishes at $g_3(r_1, r_2, \pi/3)$.
By contrast, the clustered structure has a delta peak at $g_3(r_1, r_2, \pi/3)$ and vanishes at $g_3(r_1, r_2, \pi)$.

\begin{figure*}[htp]
    \centering
    \subfloat[]{\includegraphics[width = 100mm]{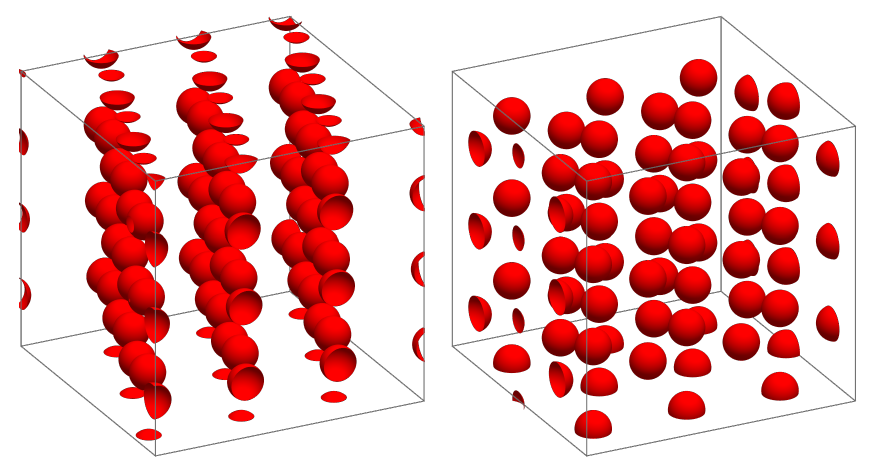}}
    \subfloat[]{\includegraphics[width = 60mm]{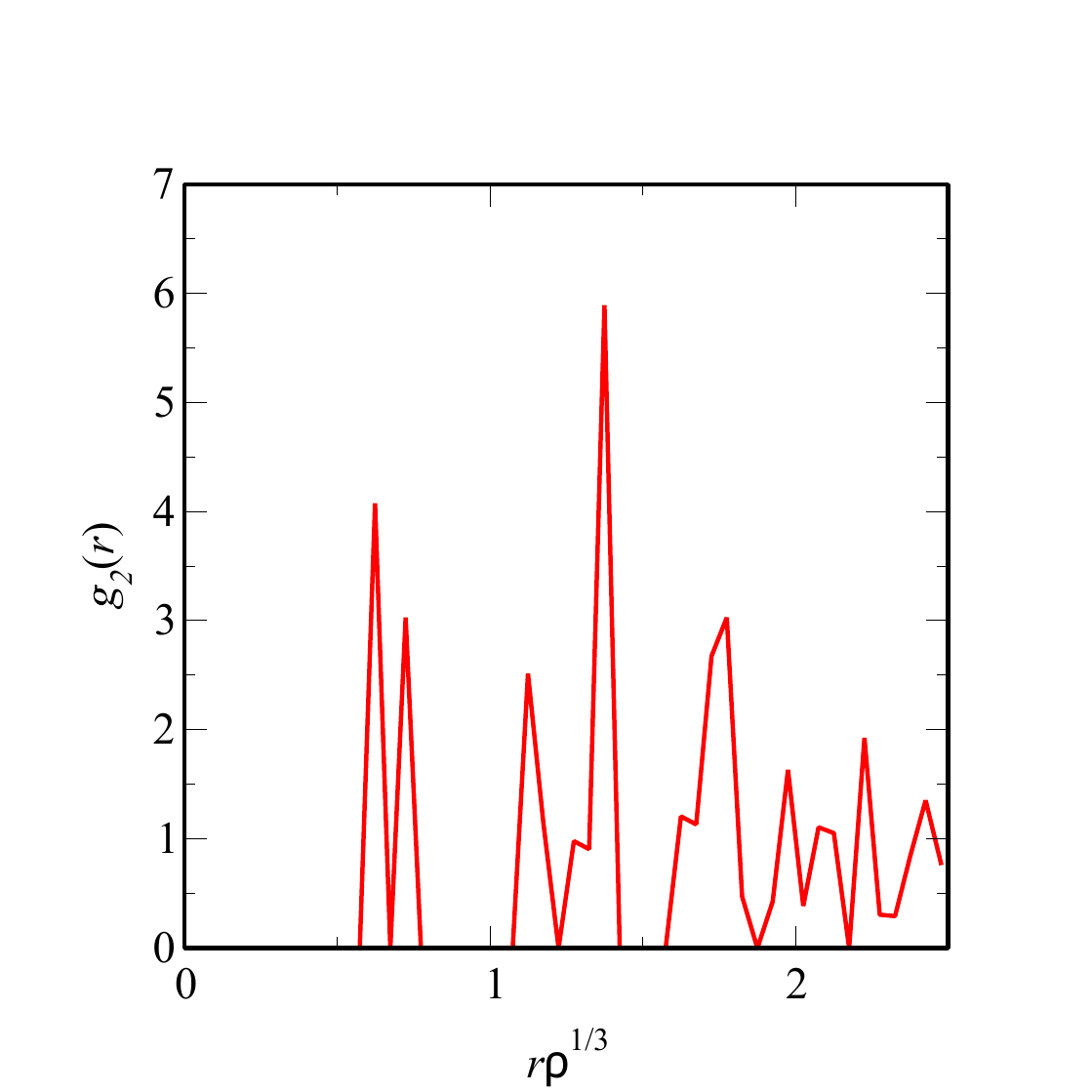}}
    \caption{(a) A pair of 3D inequivalent isospectral crystals with 2-particle bases, identified via our algorithm using deformable fundamental cells.
    The fundamental cells are triclinic.
    (b) Radial distribution function of both crystals shown in (a), plotted with bin size $0.05\rho^{-1/3}$.}
    \label{fig:3D}
\end{figure*}

\subsection{3D Cases}
Using deformable fundamental cells, our algorithm identifies pairs of 3D inequivalent isospectral crystals with $n = 2$, and an example with triclinic fundamental cells are shown in Fig. \ref{fig:3D}(a), whose $g_2(r)$ is plotted in Fig. \ref{fig:3D}(b).
No 3D isospectral crystals with $n = 1$ are found, i.e., $n_{\text{min}}(3) = 2$ via our numerical methods.
The result for $n_{\text{min}}(3)$ obtained via our algorithm is consistent with Schiemann's proof that there exist no 3D inequivalent isospectral lattices \cite{Sc90b}, again attesting to the accuracy of the algorithm.

We have also performed the optimization procedure in Sec. \ref{sec:alg} using fixed cubic and tetragonal fundamental cells.
In these cases, we find that one requires at least $n = 3$ to achieve inequivalent isospectral crystals of such high symmetries, which is larger than  $n_{\text{min}}(3) = 2$ using deformable fundamental cells.
The larger value of $n$ needed for cubic and tetragonal fundamental cells is due to the fact that the required symmetries imposes significant constraints on the fundamental-cell vectors, which reduces the number of degrees of freedom that can be employed for two inequivalent crystals to match their pair statistics.
In general, we have
\begin{equation}
    n_{\text{min}}(d) = \min_{F}n'_{\text{min}}(d; F),
    \label{nmindF}
\end{equation}
where $n'_{\text{min}}(d; F)$ is the minimum $n$ for isospectrality using a fixed fundamental cell $F$. 
Equation (\ref{nmindF}) implies that $n'_{\text{min}}(d; F) \geq n_{\text{min}}(d)$, and thus using fixed fundamental cells does not affect our conclusions about $n_{\text{min}}(d)$.

As in the 2D case, the isospectral 2-particle bases shown in Fig. \ref{fig:3D}(a) clearly possess different higher-order correlation functions, as the structure on the left contains chain-like motifs, whereas the one on the right contains well-separated individual particles.
In this case, both $g_3$ and $g_4$ of the two crystals are distinctly different.
Chain-like motifs lead to a prevalence of approximately linear 3- and 4-point configurations, whereas well-separated individual particles lead to equilateral-triangle-like 3-point configurations and regular-tetrahedron-like 4-point configurations.

\begin{figure*}[htp]
    \centering
    \subfloat[]{\includegraphics[width = 70mm]{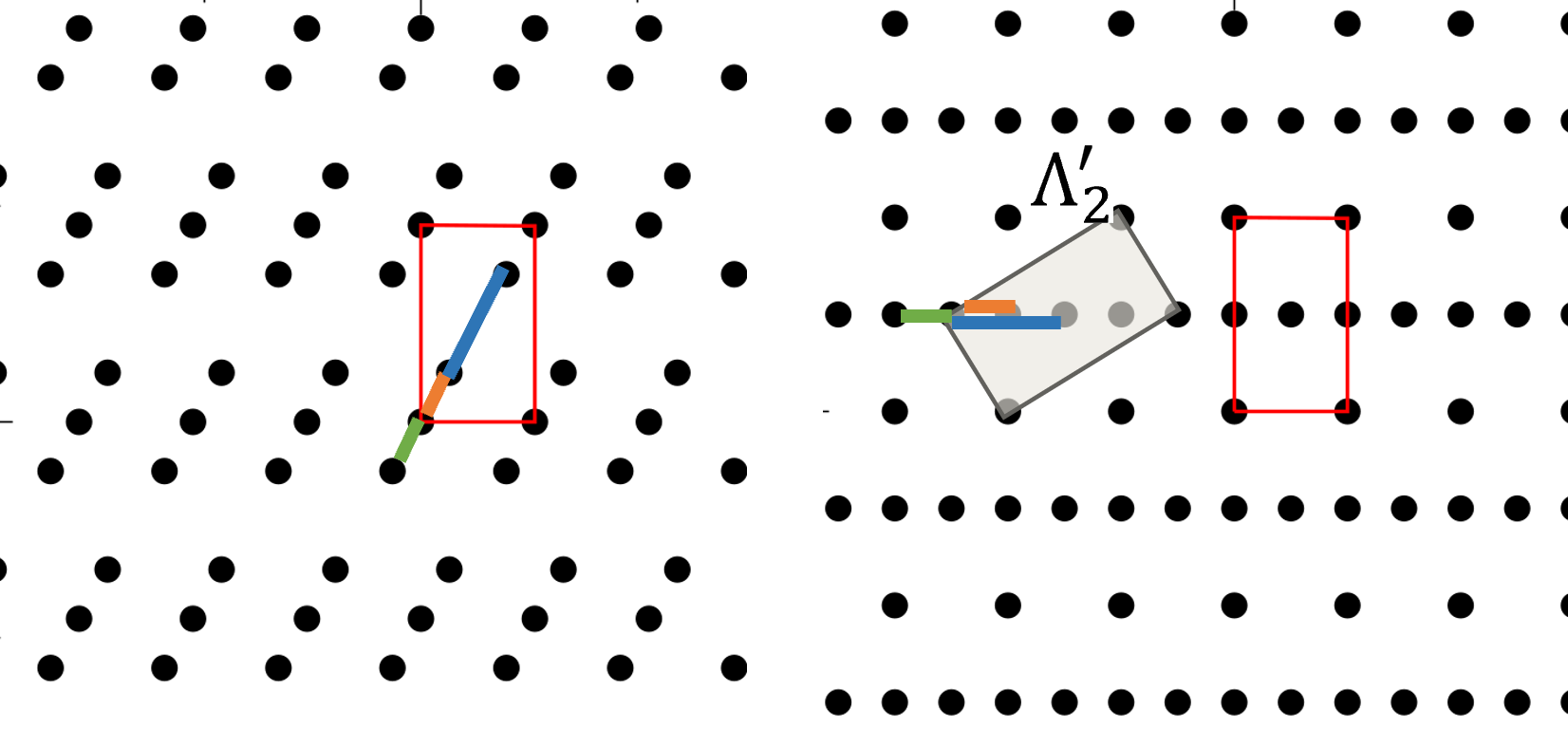}}
    \hspace{10mm}
    \subfloat[]{\includegraphics[width = 70mm]{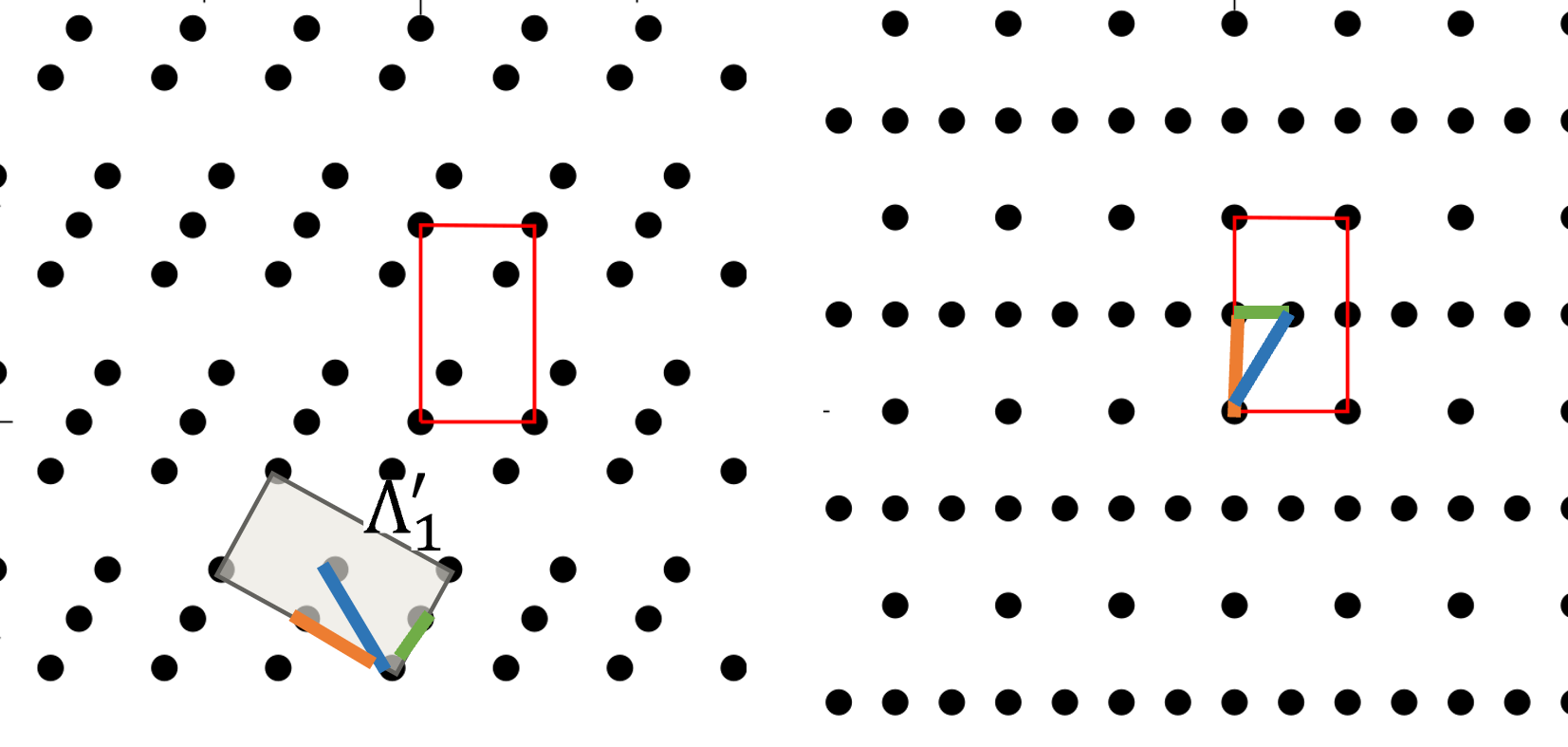}}
    \caption{(a) Illustration of the proof that the theta series of the crystals in Fig. \ref{fig:2D}(c) satisfy $\Theta_{C_1} \subseteq \Theta_{C_2}$.
    Lines of the same color indicate pairs of $\mathbf{r}_{ij}$ and $\mathbf{r}_{ij}'$. 
    The corners of the gray rectangle indicate four points in $\Lambda_2'$.
    (b) Illustration of the proof that the theta series of the crystals in Fig. \ref{fig:2D}(c) satisfy $\Theta_{C_2} \subseteq \Theta_{C_1}$.
    The corners of the gray rectangle indicate four points in $\Lambda_1'$.
    Taken together, (a) and (b) show that the two crystals are isospectral for all pair distances.}
    \label{fig:2Dproof}
\end{figure*}

\subsection{Rigorous Results on Isospectrality for All Pair Distances}
To search for inequivalent isospectral crystals, our algorithm considers the radial distribution functions and theta series only for a finite range of radial distances, i.e., up to twice the longest diagonal of the fundamental cell. 
Therefore, it is important to investigate whether the 2D and 3D cases identified above are isospectral for all $r$.
For this purpose, here we introduce a theorem that enables one to conclude isospectrality for all $r$ in certain cases, which include all of the 2D 3-particle bases in Fig. \ref{fig:2D}.

\begin{thm}
    Let $C_1$ and $C_2$ be two $d$-dimensional crystals of $n$-particle basis with common underlying basis vectors $\mathbf{a}_1, \dots, \mathbf{a}_d$.
    One has $\Theta_{C_1} \subseteq \Theta_{C_2}$, if for any pair displacement vector $\mathbf{r}_{ij} = \mathbf{p}_i - \mathbf{p}_{j} = \lambda_1 \mathbf{a}_1 + \dots + \lambda_d \mathbf{a}_d$ between particles $\mathbf{p}_i, \mathbf{p}_j \in \mathbf{P}_{C_1}$,
    there exist vectors  $\mathbf{a}'_1, \dots, \mathbf{a}'_d$ and some $\mathbf{q}\in C_2$, such that the following conditions are satisfied:
    \begin{enumerate}
        \item There exists an isometry $T\in E(d)$, such that $\mathbf{a}'_k = T(\mathbf{a}_k)$ for each $k = 1, \dots, d$;
        \item The set $\Lambda'_2 \coloneqq \{n_1 \mathbf{a}'_1 + \dots n_d\mathbf{a}'_d + \mathbf{q}: n_1, \dots, n_d \in \mathbb{Z}\}$ satisfies $\Lambda'_2 \subseteq C_2$; and
        \item The corresponding vector  $\mathbf{r}'_{ij} \coloneqq \lambda_1 \mathbf{a}'_1 + \dots + \lambda_d \mathbf{a}'_d$ satisfies $\mathbf{r}'_{ij} + \mathbf{q} \in C_2$.
    \end{enumerate}
    We remark that by definition [Eq. (\ref{thetaC})], $C_1$ and $C_2$ are isospectral if $\Theta_{C_1} \subseteq \Theta_{C_2}$ and $\Theta_{C_2} \subseteq \Theta_{C_1}$.
    \label{theorem}
\end{thm}

\begin{proof}
    Any pair displacement vector $\mathbf{r}_1$ between two particles in $C_1$ can be expressed as 
    \begin{equation}
        \mathbf{r}_1 = \mathbf{r}_{ij} + n_1\mathbf{a}_1 + \dots + n_d\mathbf{a}_d = (\lambda_1 + n_1) \mathbf{a}_1 + \dots + (\lambda_d + n_d) \mathbf{a}_d
    \end{equation}
    for some $\mathbf{r}_{ij}$ between $\mathbf{p}_i, \mathbf{p}_j \in \mathbf{P}_{C_1}$ and integers $n_1, \dots, n_d$.
    Consider the vector
    \begin{equation}
    \begin{split}
        \mathbf{r}_2 &= (\lambda_1 + n_1) \mathbf{a}'_1 + \dots + (\lambda_d + n_d) \mathbf{a}'_d \\
        &= \mathbf{r}'_{ij} + n_1\mathbf{a}'_1 + \dots n_d\mathbf{a}'_d\\
        &= (\mathbf{r}'_{ij} + \mathbf{q}) - (-n_1 \mathbf{a}'_1 - \dots - n_d\mathbf{a}'_d + \mathbf{q}). 
    \end{split}
    \end{equation}
    Clearly, $\mathbf{r}_2$ is a pair distance vector between particles $(\mathbf{r}'_{ij} + \mathbf{q})$ and $(-n_1 \mathbf{a}'_1 - \dots - n_d\mathbf{a}'_d + \mathbf{q})$ in $C_2$.
    Because $\mathbf{a}'_k = T(\mathbf{a}_k)$ for each $k$, we have $|\mathbf{r}_1| = |\mathbf{r}_2| \in \Theta_{C_2}$ for any $|\mathbf{r}_1| \in \Theta_{C_1}$.
\end{proof}

We illustrate in Fig. \ref{fig:2Dproof}(a)  that the pair of crystals in Fig. \ref{fig:2D}(e) satisfy $\Theta_{C_1} \subseteq \Theta_{C_2}$, where $C_1$ and $C_2$ refer to the pgg configuration on the left and pmm configuration on the right, respectively. 
Here, thick line segments of the same color indicate pairs of $\mathbf{r}_{ij}$ and $\mathbf{r}_{ij}'$. 
The corners of the gray rectangle indicate four points in $\Lambda_2'$, whose associated $\mathbf{q}$ is given by the meeting point of the thick line segments.
Similarly, Fig. \ref{fig:2Dproof}(b) shows that $\Theta_{C_2} \subseteq \Theta_{C_1}$, where the corners of the gray rectangle now indicate four points in $\Lambda_1'$.
Taken together, Fig. \ref{fig:2Dproof}(a) and (b) show that the two crystals are isospectral for all pair distances.

Similarly, using this theorem, one can prove isospectrality for all $r$ for the 2D 3-particle bases in Fig. \ref{fig:2D_full}(a) and Fig. \ref{fig:2D}(a) and (c).
However, we find it challenging to check for the conditions of this theorem in the case of the 3D 2-particle bases in Fig. \ref{fig:3D} due to the difficulty of extracting from the crystal configurations the regions corresponding to $\Lambda'_1$ and $\Lambda'_2$.
We remark that the converse of Theorem \ref{theorem} may not hold, i.e., isospectral crystals may not satisfy the conditions this theorem.

\section{Many-Body System with Degenerate Crystalline Ground States}
\label{sec:ground}

The identification of isospectral crystals enables one to study the degeneracy of the ground-state manifold under isotropic pair potentials.
Here, using inverse statistical-mechanical techniques, we show that there exist many-body systems under isotropic pair potentials whose low-temperature configurations can lead to both of two isospectral crystals structures.
We illustrate such ground-state degeneracy with the example in Fig. \ref{fig:2D}(a).

To search for degenerate crystalline ground-state manifolds under the action of isotropic potentials, we use the inverse technique developed in Ref. \cite{Re06a}, which determines radial pair interactions that yield ground states with targeted radial distribution functions.
This method uses a parametrized pair potential $v(r ; \mathbf{b})$, and optimizes over the potential parameters $\mathbf{b}$ to find the potential that leads to the most robust and defect-free self-assembly of the structure with the targeted $g_2(r)$ at a given $\rho$.
Here, we set the targeted $g_2(r)$ to be the one plotted in Fig. \ref{fig:2D}(b), corresponding to the 2D isospectral  3-particle bases with rhombic fundamental cells [Fig. \ref{fig:2D}(a)].
This target is chosen because of the high symmetry of both crystals, which reduces the number of peaks in $g_2(r)$ and simplifies the optimization of the potential.

To proceed, we prescribe functional form of the isotropic pair potential based on our prior knowledge of the forms of pair potentials that yield other unusual ground states, such as the honeycomb lattice \cite{Re06a}. 
We use the form
\begin{equation}
\begin{split}
    \beta v(r;\mathbf{b}) &= \left(\frac{r_0}{r}\right)^{12} + \varepsilon_1\exp\left[\left(\frac{r}{\sigma_1}\right)^2\right] \\
    &-\sum_{j = 2}^6 \varepsilon_j\exp\left[\left(\frac{r-r_j}{\sigma_j}\right)^2\right],
\end{split}
    \label{pot}
\end{equation}
where $\beta = 1/(k_BT)$, $T$ is the temperature, $\mathbf{b}$ is a vector that represents the components of all potential parameters $r_j, \epsilon_j$ and $\sigma_j$ subject to optimization.
The first term in (\ref{pot}) is a strong short-ranged repulsion that enforces the shortest radial distance, the second term is a soft repulsive Gaussian with a relatively long interaction range, i.e., $\sigma_1 \sim 2\rho^{-1/2}$, which is designed to destabilize dense triangle-lattice-like motifs, as done in Ref. \cite{Re06a}, and the remaining terms are narrow wells added to generate the desired peaks in the targeted $g_2(r)$.
Using the BFGS algorithm \cite{Liu89}, we vary the potential parameters $\mathbf{b}$ to minimize an objective function that corresponds to the difference between the lattice sum (i.e., energy per particle) computed from the targeted $g_2(r)$ and those for competitor crystal structures:
\begin{equation}
    \Delta E = \epsilon^T(\rho) - \min_{C \in \mathcal{C}, \rho' \in (\rho, \infty)}\epsilon^C(\rho'),
\end{equation}
where $\epsilon^T(\rho)$ is the energy per particle under $v(r; \mathbf{b})$ computed from targeted $g_2(r)$ at number density $\rho$, $\mathcal{C}$ is the collection of the competitor crystal structures, and $\epsilon^C(\rho')$ is the energy per particle for the competitor structure $C$ at density $\rho'$.
The initial competitor set $\mathcal{C}$ include triangle and square lattices, as well as Kagome and honeycomb crystals.
Once the optimized parameters $\mathbf{b}$ are obtained, we perform simulated annealing under $v(r; \mathbf{b})$ to find the low-temperature states, which are usually crystals with defects. 
If $g_2(r)$ for the resulting low-temperature structure is significantly different from the target, we add this structure into the competitor set $\mathcal{C}$ and repeat the BFGS minimization of $\Delta E$.
The procedure iterates until the radial distances of first four maxima in the simulated $g_2(r)$ match those of the targeted $g_2(r)$ within errors $\pm 0.025\rho^{-1/2}$.
Note that the iterative update of $\mathcal{C}$ is not included in the original algorithm \cite{Re06a}, because the targeted structures in Ref. \cite{Re06a}, e.g., the honeycomb lattice, do not have many close competitors. 
By contrast, for our targeted $g_2(r)$, 5 iterations of updates on $\mathcal{C}$ were required to achieve convergence of the potential.

\begin{figure*}[htp]
    \centering
    \subfloat[]{\includegraphics[width = 60mm]{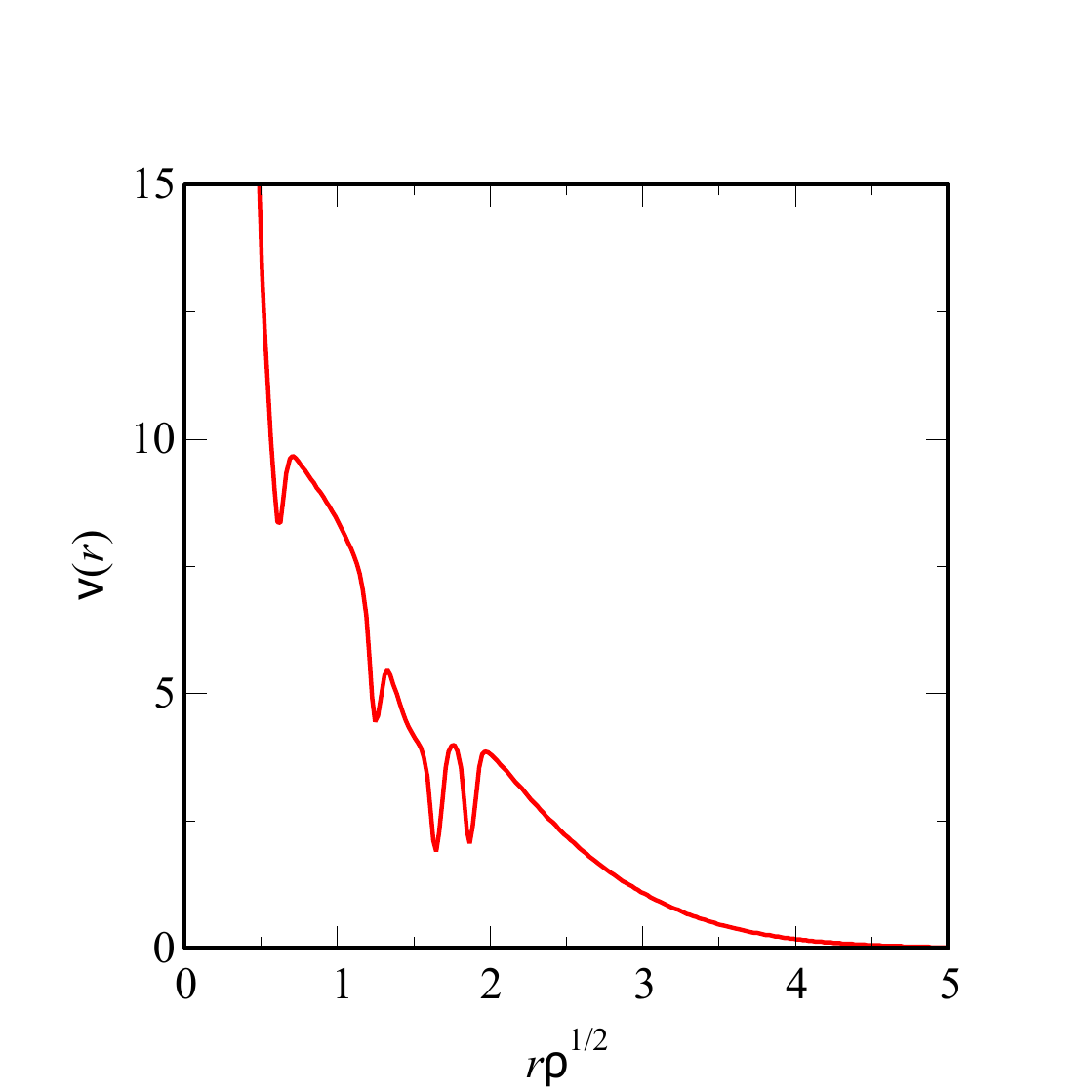}}
    \subfloat[]{\includegraphics[width = 60mm]{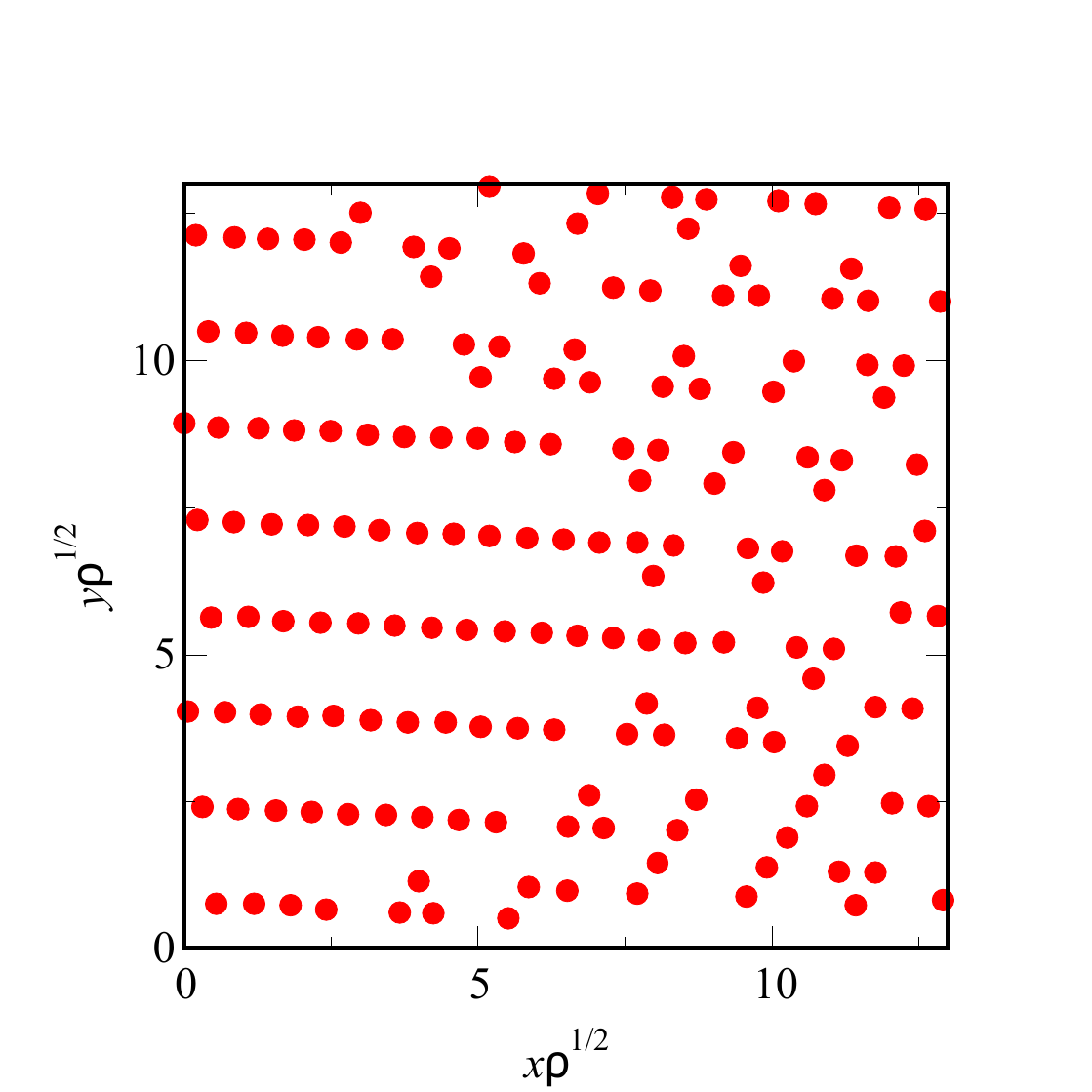}}
    \subfloat[]{\includegraphics[width = 60mm]{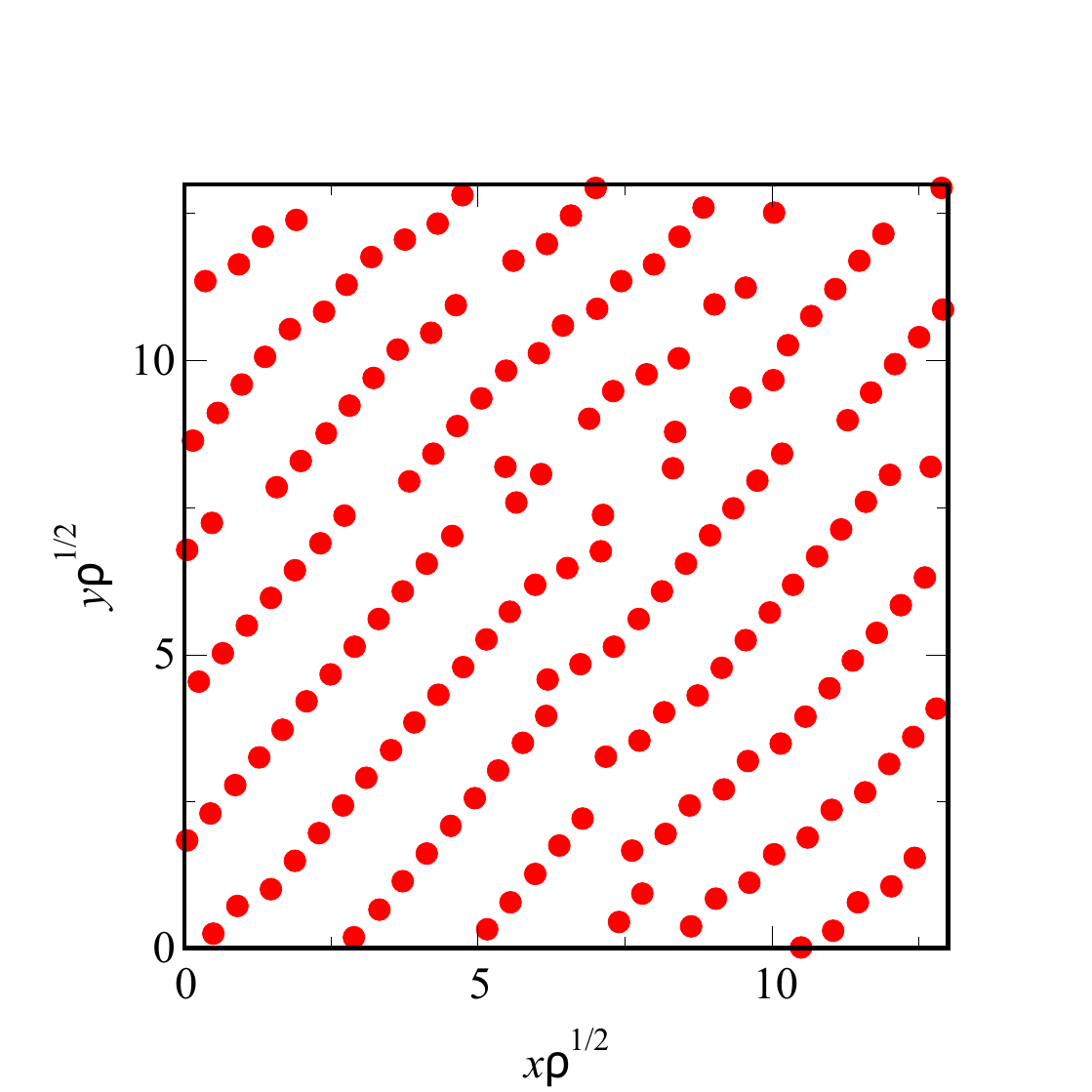}}
    \caption{(a) Pair potential (\ref{pot}) designed to yield degenerate ground states given by the isospectral crystals with hexagonal fundamental cells shown in Fig. \ref{fig:2D}(a), which we have found using the inverse technique in Ref. \cite{Re06a}.
    (b) A low-temperature configuration obtained via simulation annealing under the potential shown in (a) with $N = 168$ particles and cooling rate $\gamma = 0.999$, with $10^5$ sweeps at each temperature.
    (c) Same as (b), but with a faster cooling rate $\gamma = 0.997$. }
    \label{fig:pot}
\end{figure*}

Using the inverse technique described above, we have found an isotropic pair potential of form (\ref{pot}) whose 2D low-temperature configurations obtained via simulated annealing can lead to both of two inequivalent isospectral crystal structures with 3-particle bases, shown in Fig. \ref{fig:2D}(a).
The optimized potential [Fig. \ref{fig:pot}(a)] has a repulsive part on the range $r \leq 5\rho^{-1/2}$, but contains 4 sharp local minima at radial distances corresponding to the first 4 peaks in the targeted $g_2(r)$ shown in Fig. \ref{fig:2D}(b).
The form (\ref{pot}), characterized by soft repulsive interaction on the small to intermediate range, as well as sharp attractive wells at specific pair distances, typically appear systems of DNA-grafted nanoparticles \cite{Ml13, Ku17} or polymers \cite{Ya04}.

Importantly, we find that the proportion of the two inequivalent isospectral configurations can be controlled by the cooling rate $\gamma$ in the simulated annealing procedure.
At a very slow cooling rate $\gamma = 0.999$, the two degenerate ground-state structures, characterized by linear chains and equilateral triangular clusters, occur with approximately equal probability [Fig. \ref{fig:pot}(b)].
However, at a faster cooling rate $\gamma = 0.997$, linear chains occur with a much higher probability than triangular clusters [Fig. \ref{fig:pot}(c)].
Indeed, almost all particles are part of the striped structure. 
The fact that chains are dynamically favored over equilateral triangular clusters is expected, because at positive temperatures, chain-like structures can bend or glide in parallel to one another without a significant energy cost, whereas the triangular clusters form part of a crystal with 3-fold rotational symmetry and cannot undergo large-scale collective motion without significantly increasing the configurational potential energy.
This example suggests the possibility to use inverse techniques \cite{To09a, Co09d, Ma16b} to study how different cooling rates can preferentially lead to different ground states (under the action of isotropic pair potentials) characterized by inequivalent isospectral crystals.
Note that varying the cooling rate has already been a common practice to exert control over crystal polymorphism and morphology \cite{Ni08, Li20b}.

\section{Conclusions and Discussion}
\label{sec:conc}
In this work, we have used numerical and rigorous methods to find $n_{\text{min}}(d)$, the minimum value of $n$ for inequivalent crystals with the same theta function in $\mathbb{R}^d$.
Using a precise numerical algorithm, we have identified inequivalent isospectral 4-, 3-, and 2-particle bases in one, two and three dimensions, respectively. 
In 1D and 3D cases, the values of $n_{\text{min}}(d)$ via our numerical procedures are consistent with rigorous results, demonstrating the reliability and precision of the algorithm. 
In the 2D case, while it is challenging to prove that $n_{\text{min}}(2) > 2$, we have rigorously shown that many 2D inequivalent isospectral 3-particle bases found via the algorithm indeed possess identical theta series up to infinite pair distances.
In summary, we conjecture that
$n_{\text{min}}(d) = 4, 3, 2$ for $d = 1, 2, 3$, respectively.

Our finding that $n_{\text{min}}(d)$ decreases as $d$ increases is consistent with the decorrelation principle \cite{To06b}, which states that spatial correlations that exist for a particular model in lower dimensions diminish as the space dimension becomes larger.
As a result, degeneracy of pair correlation functions becomes increasingly prevalent in higher dimensions.
The decorrelation principle has been shown in various cases of disordered many-body systems \cite{To06b, Wa22b}, and our work shows that it is also relevant for ordered states, i.e., crystals.
It is important to remark that $n_{\text{min}}(4) = 1$ because an example of 4D inequivalent isospectral lattices have been identified by Schiemann \cite{Sc90}.
Thus, it is likely that $n_{\text{min}}(d) = 1$ for all $d \geq 4$, because inequivalent isospectral crystals in higher dimensions can be constructed from layers of crystals in lower dimensions, and isospectrality of the layers isospectrality would imply the isospectrality of the higher-dimensional crystals.

The existence of inequivalent isospectral crystals implies that the ground-state manifold for certain many-body systems under isotropic pair interactions can lead to degenerate crystalline configurations.
Using inverse statistical-mechanical methods, we show that there indeed exist systems under isotropic pair interactions whose low-temperature configurations can lead to both of two isospectral crystals structures, the proportion of which can be controlled via the cooling rate.
Thus, our findings provide general insights into the structural and ground-state degeneracies of crystal structures  as determined by radial pair statistics and radial pair interactions, respectively.

Our work motivates the study of the enumeration and classification of crystal structures with an $n$-particle basis where $n$ is greater than or equal to $n_\text{min}(d)$ in space dimension $d$. 
One could begin to attack this problem by using our algorithm with $n = n_\text{min}(d) + 1, n_\text{min}(d) + 2, \dots$, etc.
In these calculations, the upper cut-off radius $R$ in (\ref{dg2}) should be carefully chosen such that the number of pair distances in $(0, R)$ is much larger than $n_F$, as discussed in Sec. \ref{sec:alg}.
We expect that as $n$ increases beyond $n_{\text{min}}(d)$ for an arbitrary but fixed space dimension $d$, there may exist cases in which the number of mutually inequivalent crystals, $\mathcal{N} \geq 3$, possess identical theta series. 
(Note that $\mathcal{N} = 2$ for all isospectral cases identified in this work, since our algorithm aims to find \textit{pairs} of inequivalent isospectral crystals.)
One could attempt to identify triplets, quadruplets, etc., of inequivalent isospectral crystals by using our algorithm with a generalized pair-distance metric $D_{g_2}(C_1, \dots, C_\mathcal{N})$ that vanishes if and only if $g_2(r)$ for all crystals match on $(0, R)$, as well as a generalized geometric distance metric $\xi(C_1, \dots, C_\mathcal{N})$ that vanishes if and only if any two of the crystals $C_1, \dots, C_\mathcal{N}$ are equivalent.
A promising avenue for future research is to determine the maximum $\mathcal{N}$ for given $d$ and $n \geq n_{\text{min}}(d)$.

Finally, we remark that the existence of inequivalent isospectral crystals provides important perspective regarding the interpretation of crystallographic and diffraction data.
For example, the powder X-ray diffraction pattern of a crystal is equivalent to the angular averaged structure factor $S(k) = 1 + \rho\int_{\mathbb{R}^d} [g_2(r) - 1] d \mathbf{r}$ \cite{El11, Ho19}.
Thus, given the possibility of inequivalent isospectral crystals, the structures corresponding to a given diffraction pattern can be non-unique, and one requires other information than the angular averaged pair function to fully determine the crystal structure.

\section*{Acknowledgements}
This work is supported by the National Science Foundation CBET-2133179.

\section*{Appendix}
\subsection{Theta series for some 2D crystals}
To get an intuitive idea of the information encoded in the theta series (\ref{theta_def}), we explicitly show the first several terms of the theta series for some well-known 2D Bravais lattices ($n = 1$) and non-Bravais lattices (crystals) ($n \geq 2$).
For Bravais lattices, the pair distances are measured in units of a particle's nearest-neighbor distance.
For non-Bravais crystals, the distances are in units of the nearest-neighbor distance of the underlying Bravais lattice.

For the square lattice ($n = 1$) \cite{Co93}:
\begin{equation}
    \Theta(q) = 1 + 4q^1 + 4q^2 + 4q^4 + 8q^5 + \dots.
\end{equation}
Thus, in the square lattice, there are 4 pair displacement vectors with squared norm 1, 4 with squared norm 2, etc.
The theta series of other crystals are similarly interpreted.

For the triangle lattice ($n = 1$) \cite{Co93}:
\begin{equation}
    \Theta(q) = 1 + 6q^1 + 6q^3 + 6q^4 + 12 q^7 + \dots.
\end{equation}

For the honeycomb crystal ($n = 2$) \cite{Co93}:
\begin{equation}
    \Theta(q) = 1 + 3q^{1/3} + 6q^{1} + 3q^{4/3} + 6q^{7/3} + \dots.
\end{equation}

For the Kagom{\'e} crystal ($n = 3$) \cite{Co93}:
\begin{equation}
    \Theta(q) = 1 + 4q^{1/4} + 4q^{3/4} + 6q^{1} + 8q^{7/4} + \dots.
\end{equation}

\subsection{Parametrization of Crystals}
%A crystal $C$ such that $\mathbf{0}\in C$
Here, we describe our methods to parametrize 2D and 3D crystals, i.e., to eliminate the $d(d-1)/2 + 1$ degrees of freedom due to rotations and scaling, thereby transforming the vectors $\mathbf{a}_i, \dots, \mathbf{a}_d$ and $\mathbf{p}_2^{(i)}, \dots, \mathbf{p}_n^{(i)}$ into $n_F$ free scalar parameters subject to optimization in our algorithm described in Sec. \ref{sec:alg}.
Table \ref{tab:2dparametrization} shows our parametrization for a 2D crystal, where $c_i$'s are the free parameters subject to optimization.
Similarly, Table \ref{tab:3dparametrization} shows our parametrization for a 3D crystal.
To search for inequivalent isospectral 2D and 3D crystals, we have assumed that the pair of crystals $C_1, C_2$ have identical underlying lattice, and thus they share common values of $c_1, c_2$ if $d = 2$, and common values of $c_1, \dots, c_5$ if $d = 3$.
Table \ref{tab:crystalparams} presents the optimized parameters for the 2D and 3D inequivalent isospectral crystals identified via our algorithm, whose configurations are shown in Figs. \ref{fig:2D_full}, \ref{fig:2D} and \ref{fig:3D}.

\begin{table*}[htp]
    \centering
    \caption{Parametrization method used in this work to transform the basis and particle-position vectors for a 2D crystal into free parameters subject to optimization.}
    \begin{tabular}{||c|p{10cm}|c||}
        Parameter & Definition & Range\\
        \hline
        $c_1$ & $|\mathbf{a}_2|/|\mathbf{a}_1|$ & $[1, 5]$\\
        $c_2$ & $\mathbf{a}_1\cdot\mathbf{a}_2/|\mathbf{a}_1|^2$, the projected length of $\mathbf{a}_2$ on $\mathbf{a}_1$ divided by $|\mathbf{a}_1|$. & $[0, 1/2]$ \\
        $c_{1 + j}, c_{2 + j}$ & Components of $[\mathbf{a}_1, \mathbf{a}_2]^{-1}\mathbf{p}_j$. & $[0, 1)$\\
    \end{tabular}
    \label{tab:2dparametrization}
\end{table*}

\begin{table*}[htp]
    \centering
    \caption{Parametrization method used in this work to transform the basis and particle-position vectors for a 3D crystal into free parameters subject to optimization.}
    \begin{tabular}{||c|p{10cm}|c||}
        Parameter & Definition & Range\\
        \hline
        $c_1$ & $|\mathbf{a}_2|/|\mathbf{a}_1|$ & $(0, 1]$\\
        $c_2$ & $|\mathbf{a}_3|/|\mathbf{a}_2|$ & $(0, 1]$\\
        $c_3$ & Angle between $\mathbf{a}_1$ and $\mathbf{a}_2$, & $[0, \pi/2]$ \\
        $c_4$ & The ratio $c_4'/c_3$, where $c_4'$ is the angle between $\mathbf{a}_1$ and the projected vector of $\mathbf{a}_3$ on the plane containing $\mathbf{a}_1$ and $\mathbf{a}_2$. & [0, 1] \\
        $c_5$ & Angle between $\mathbf{a}_3$ to the plane containing $\mathbf{a}_1$ and $\mathbf{a}_2$. & $[0, \pi/2]$\\
        $c_{4 + j}, c_{5 + j}, c_{6 + j}$ & Components of $[\mathbf{a}_1, \mathbf{a}_2, \mathbf{a}_3]^{-1}\mathbf{p}_j$. & $[0, 1)$\\
    \end{tabular}
    \label{tab:3dparametrization}
\end{table*}

\begin{table*}[htp]
    \centering
    \caption{Optimized parameters for the 2D and 3D inequivalent isospectral crystals identified via our algorithm.}
    \begin{tabular}{||c|c|c|c|c||c|c||}
       Parameter for $d=2$ & Fig. \ref{fig:2D_full}(a) & Fig. \ref{fig:2D}(a) & Fig. \ref{fig:2D}(c) & Fig. \ref{fig:2D}(e) & Parameter for $d=3$ &  Fig. \ref{fig:3D}(a) \\
       \hline
        $c_1$    &  $\left(\sqrt[4]{55} \sqrt{\frac{3 \sqrt{55}}{4}+\frac{3}{4 \sqrt{55}}}\right)/(2 \sqrt{3})$ & 1   & 1 & $\sqrt{3}$ & $c_1$ & 0.9968\\ 
        $c_2$    &  1/4 & 1/2 & 0 & 0 & $c_2$ & 0.8087\\ 
        $c_3^{(1)}$ &  3/4 & 0   & 0 & 1/4 & $c_3$ & 1.264\\ 
        $c_4^{(1)}$ & 1/4 & 1/3 & $1/(1 + \sqrt{3})$ & 1/4 & $c_4$ & 0.4268 \\ 
        $c_5^{(1)}$ & 1/4 & 0   & $1/(2 + 2\sqrt{3})$& 3/4 & $c_5$ & 1.542 \\ 
        $c_6^{(1)}$ & 3/4 & 2/3 & $(1 + \sqrt{3}/2)/(1 + \sqrt{3})$ & 3/4 & $c_6^{(1)}$ & 0.03015\\ 
        $c_3^{(2)}$ & 3/4 & 0 & 0 & 0 & $c_7^{(1)}$ & 0.5288\\ 
        $c_4^{(2)}$ & 1/2 & 1/3 & $\sqrt{3}/(1 + \sqrt{3})$ & 1/2 & $c_8^{(1)}$ & 0.9164\\ 
        $c_5^{(2)}$ & 1/2 & 2/3 & $(1 + \sqrt{3}/2)/(1 + \sqrt{3})$ & 1/2 & $c_6^{(2)}$ & 0.7093\\ 
        $c_6^{(2)}$ & 0 & 1/3 & $(\sqrt{3} + 1/2)/(1 + \sqrt{3})$ & 1/2 & $c_7^{(2)}$ & 0.2096\\ 
        & & & & & $c_8^{(2)}$ & 0.5583\\
    \end{tabular}
    \label{tab:crystalparams}
\end{table*}

\subsection{Parameters of the Isotropic Pair Potential with Degenerate Crystalline Ground States}

Table \ref{tab:potparams} presents the optimized parameters in the isotropic pair potential (\ref{pot}), whose low-temperature configurations can lead to both of two isospectral crystals.

\begin{table*}[htp]
    \centering
    \caption{Optimized parameters in the isotropic pair potential (\ref{pot}), whose low-temperature states contain configurations of both of two isospectral crystals.
    All length parameterss $r_j, \sigma_j$ are made dimensionless in units of $\rho^{-1/2}$, and all energy parameters $\varepsilon_j$ are made dimensionless in units of $k_BT$.}
    \begin{tabular}{||c|c||c|c||c|c||}
       $r_0$ & 0.6146 & $\sigma_2$ & 0.05 & $r_2$ & 0.6146 \\
       $\varepsilon_1$ & 11.00 & $\sigma_3$ & 0.2906 & $r_3$ & 1.507 \\
       $\sigma_1$ & 1.972 & $\sigma_4$ & 0.05 & $r_4$ & 1.249 \\
        $\varepsilon_2, \dots, \varepsilon_6$ & 2 & $\sigma_5$ & 0.05 & $r_5$ & 1.647 \\
        & & $\sigma_6$ & 0.05 & $r_6$ & 1.868 
    \end{tabular}
    \label{tab:potparams}
\end{table*}

\clearpage

\end{document}